\documentclass[11pt]{article} 
\usepackage{fullpage,graphicx,amsmath,amsfonts,amssymb,amsthm,dsfont,authblk}
\usepackage[utf8]{inputenc} 
\usepackage[colorlinks]{hyperref} 
\usepackage{bm} 
\usepackage{physics} 
\usepackage{cleveref} 
\usepackage{algorithm} 
\usepackage{algpseudocode} 
\usepackage{xcolor} 
\usepackage{dsfont} 
\usepackage{mathtools}
\usepackage{booktabs}
\usepackage{qcircuit}

\DeclareMathOperator{\maj}{maj}

\theoremstyle{plain}

\theoremstyle{definition}
\newtheorem{defn}{Definition}
\newtheorem{prob}{Problem}

\theoremstyle{remark}

\newtheorem{theorem}{Theorem}[section]
\newtheorem{corollary}{Corollary}[theorem]
\newtheorem{lemma}[theorem]{Lemma}

\theoremstyle{definition}

\theoremstyle{remark}
\newcommand{\eq}[1]{Eq.~\hyperref[eq:#1]{(\ref*{eq:#1})}}
\renewcommand{\sec}[1]{\hyperref[sec:#1]{Section~\ref*{sec:#1}}}
\newcommand{\app}[1]{\hyperref[app:#1]{Appendix~\ref*{app:#1}}}
\newcommand{\tab}[1]{\hyperref[tab:#1]{Table~\ref*{tab:#1}}}
\newcommand{\fig}[1]{\hyperref[fig:#1]{Figure~\ref*{fig:#1}}}
\newcommand{\figa}[2]{\hyperref[fig:#1]{Figure~\ref*{fig:#1}#2}}
\newcommand{\figx}[2]{\hyperref[fig:#1]{Figure~\ref*{fig:#1}(#2)}}
\newcommand{\cor}[1]{\hyperref[cor:#1]{Corollary~\ref*{cor:#1}}}
\newcommand{\alg}[1]{\hyperref[alg:#1]{Algorithm~\ref*{alg:#1}}}

\title{Unitary Quantum Cellular Automata for Density Classification}

\author[1,4]{Pedro C. S. Costa\thanks{Email: pcs.costa@protonmail.edu} }
\author[2]{Yuval R.~Sanders}
\author[3]{Pedro Paulo Balbi}
\author[4]{Gavin K. Brennen}

\affil[1]{\small{ContinoQuantum, Sydney, NSW 2093, AU}}
\affil[2]{\small{
Centre for Quantum Software and Information,
University of Technology Sydney, Sydney, NSW 2007, AU}}
\affil[3]{\small{Universidade Presbiteriana Mackenzie - FCI, Rua da Consolação 896, 01302-907 São Paulo, SP, Brazil}}
\affil[4]{\small{School of Mathematical and Physical Sciences,
Macquarie University, Sydney, NSW 2109, Australia}}

\date{}
\begin{document}

\maketitle	

\begin{abstract}
We investigate the density classification task (DCT)—determining the majority bit in a one-dimensional binary lattice—within a quantum cellular automaton (CA) framework. While there is no one-dimensional two-state, radius $r\geq 1$, deterministic CA with periodic
boundary conditions that solves the DCT perfectly, we explore whether a unitary quantum model can succeed. We employ the Partitioned Unitary Quantum Cellular Automaton (PUQCA), a number-conserving model, and via evolutionary search find solutions to  the DCT where the success condition is stipulated in terms of measurement probabilities rather than convergence to fixed-point configurations. Finally, we identify a classically simulable regime of the PUQCA in which rules that solve the DCT at fixed system sizes still.

\end{abstract}

\tableofcontents

\section{Introduction}

Given a cyclic, one-dimensional regular lattice where each element is assigned a binary state, determining which state occurs more frequently is a trivial task for a classical computer with global access to all elements. This task is known as the \textit{density classification task} (DCT), or the \textit{global majority problem}~\cite{Fates,PPdens,soldens}. However, if we impose the restriction that the program may only access local information—i.e., each lattice element can interact only with a small group of neighbours at any given time—the problem may even become unsolvable~\cite{PPdens}.

This extreme limitation occurs when additional constraints are introduced: \textit{locality} (local interactions only), \textit{space and time homogeneity} (all elements follow the same update rule across time), \textit{success condition} (computation halts when a predefined global agreement is reached), and \textit{deterministic synchrony} (all elements update simultaneously and deterministically). Cellular automata (CAs)~\cite{CAbook} naturally embody all of these restrictions. In this model, a local transition function \( f \) updates the state \( a_i(t) \in \Sigma \) of each cell based solely on the configuration of its local neighborhood; for example, 
$a_i(t+1) = f\left[a_{i-1}(t), a_i(t), a_{i+1}(t)\right],$ for all $i$, where \( i \) denotes the cell position in a one-dimensional (usually cyclic) lattice, \( t \in \mathbb{N} \) is the discrete time step, and \( \Sigma = \{0, 1\} \) is the binary alphabet.

All the restrictions described above are satisfied within this framework. Regarding the success condition required for DCT, the global agreement should emerge from local interactions and depend on the initial configuration. In particular, starting with a configuration of odd size and initial density \( \rho \) of 1s such that \( \rho > 1/2 \) (or \( \rho < 1/2 \)), the transition function \( f \) should lead the system to a stationary configuration of all 1s (or all 0s), typically after \( M \) time steps, where \( M \) is proportional to the system size \( L \).

However, it was proven in~\cite{NoSolclass} that no deterministic, binary CA rule with radius \( r \geq 1 \) can correctly classify all possible cyclic configurations. Consequently, various attempts have been made to identify near-optimal rules. Two of the most effective rules to date were found in~\cite{wolz2008very}, achieving approximately 89\% accuracy for \( 5 \times 10^5 \) randomly generated configurations of size \( L = 149 \), sampled from a Bernoulli distribution. Such probabilistic sampling is necessary due to the exponential growth in the number of configurations with \( L \), i.e., \( |\Sigma|^L \), and the fact that densities near \( \rho = 1/2 \)—the hardest to classify—occur more frequently in Bernoulli samples. While manually designed rules such as the Gács-Kurdyumov-Levin rule~\cite{gkl} exist, evolutionary methods, especially \textit{genetic algorithms} (GAs)~\cite{wolz2008very,mitchell1994evolving}, have yielded the best-performing rules to date, and we adopt this approach in our study as well.

With the advent of quantum computation, it is natural to ask whether quantum models can tackle the density classification problem. Although several quantum versions of cellular automata (QCA) have been proposed~\cite{QCAPerez,QCAmeyer,Watrous,Costa2018}, we found only one direct attempt to address DCT, presented in~\cite{Qgenetic}, where Meyer's QCA~\cite{QCAmeyer} was employed. While this approach may in principle scale to larger lattice sizes, the study was limited to inputs of length 5. Additionally, the formulation of inputs and the success condition in that work do not align well with standard classical treatments of DCT. In Ref~\cite{QCADenClass}, a unitary QCA was introduced to solve for density-classification with an update rule based on the known elementary Wolfram CA rule 232. In order to implement that QCA using local gates it was found necessary to work with a quasi-1D lattice. That model is not a perfect classifier, but it was shown to be useful as a way to efficiently perform measurement-free quantum error correction (MFQEC) for bit-flip channels.

An alternative approach was proposed in~\cite{wagner2024density}, where DCT was explored using a non-unitary quantum cellular automaton inspired by classical CA rules. They demonstrated that the majority voting rule can solve the DCT time scaling linearly with the system size. However, whether a unitary QCA can solve DCT remains an open question.

In this work, we address the question using a unitary quantum cellular automaton model, the \emph{Partitioned Unitary Quantum Cellular Automaton} (PUQCA), introduced in~\cite{Costa2018}. As shown there, PUQCA compactly captures several quantum walk models that are usually presented in abstract Hilbert space terms, in contrast to PUQCA, which admits a direct gate model interpretation. The framework is also well suited for simulating relativistic fermions—recovering the Dirac equation in the continuum limit~\cite{costa2021quantum}—and for providing a quantum analogue of classical partitioned cellular automata (PCA)~\cite{ToffoliMargolus,CostaM20}.
 Within the PUQCA framework, we have an additional constraint: the total excitation (i.e., the number of 1s) is preserved throughout the evolution. However, in PUQCA, since the evolution is unitary and preserves the number of excitations, the system cannot reach a stationary state. As a result, the traditional success condition for using cellular automata in density classification problems must be reformulated. We address this by redefining the success condition in terms of the probability distribution over configurations, rather than deterministic convergence.

A significant distinction between classical and quantum approaches is that, in classical methods, different initial configurations converge in various numbers of steps. In contrast, in our quantum model, the correct output for all configurations appears at the same fixed time step, at local lattice points.

In this framework, the rules of the QCA are determined by the unitary operations of the PUQCA. Unlike classical CA, which has a finite rule space, the quantum rule space is continuous (e.g., parametrised by rotation angles), making exhaustive exploration infeasible. Moreover, a rule that successfully classifies for one system size may not generalise to other sizes due to the lack of convergence. 

We employ a genetic algorithm to search for parameterised unitaries that solve the density classification task for a given lattice size. We also extend the search to identify a rule that performs well across multiple lattice sizes. However, since each grid point in the PUQCA corresponds to a qubit, the computational cost of simulating the system increases rapidly with system size, which limits our experiments to relatively small configurations.

On the other hand, we also show under which conditions the PUQCA becomes classically simulable. Within this regime, although we could not identify a rule that generalises as well as the PUQCA, which is not in the classical simulatable regime across different system sizes, we did find rules capable of solving the DCT for fixed lattice sizes.

The paper is organised as follows: \Cref{sec:background} reviews the density classification task  and the partitioned unitary quantum cellular automaton  used to address it. Section~\ref{sec:approach} presents our methodology for solving the DCT and establishes a lower bound for the fitness function. In \Cref{sec:transl_invaDct}, we leverage the translational invariance of the PUQCA to show that, in principle, there exists a quantum state that solves the task. Section~\ref{sec:res1} then introduces a first set of QCA rules—obtained via a genetic algorithm—that successfully solve the DCT. Section~\ref{sec:perf} shows how the PUQCA can be brought into an efficiently classically simulatable regime, which we then exploit to apply it to the DCT.  Finally, \Cref{sec:conc} concludes with a summary of our findings, their implications for density classification with QCAs, and directions for future work.

\section{Background}
\label{sec:background}

Our paper investigates the ability of quantum cellular automata to accomplish
the density classification task. In this section, we explain 
(\Cref{sec:background/PUQCA}) the specific kind of quantum cellular automata
we consider (so-called \emph{partitioned unitary} QCA). We also explain
(\Cref{sec:background/density_classification}) the density classification task
and how a cellular automaton can accomplish that task.

\subsection{One-dimensional partitioned unitary QCA}
\label{sec:background/PUQCA}

Intuitively, a cellular automaton consists of three components:
(1) a grid of points, 
(2) an initial labelling of each point as either `on' or `off', and 
(3) a transition function \( f \) such that, given an on/off labelling at time step \( t \in \mathbb{N} \), applying \( f \) produces a new labelling at time step \( t + 1 \). The transition function is expected to be both local and translationally invariant, according to a suitable notion of locality defined over the grid.

A \emph{quantum} cellular automaton extends this concept by allowing the point labels to exist in a global superposition of on/off states. In this work, we assume the transition function is specified by a unitary operator, which, in particular, implies that the evolution is linear and reversible.

Here, we restrict our attention to one-dimensional QCAs. That is, we consider a one-dimensional grid of $n$ points labelled by integers $k = 0, \ldots, n-1$, where each point may be in a superposition of the states \(\ket{0}\) (off) and \(\ket{1}\) (on). A complete configuration of the grid can then be represented by an \( n \)-bit string \( \bm{b} = (b_0, \dots, b_{n-1}) \), corresponding to the quantum state $\ket{\bm{b}} \coloneqq \ket{b_0} \otimes \cdots \otimes \ket{b_{n-1}}.$
The global state of the QCA at time $t$ can be written as
\begin{equation}
  \ket{\psi(t)} = \sum_{\ell = 0}^{2^n - 1} \alpha_\ell \ket{\ell},
\end{equation}
where \( \alpha_\ell \in \mathbb{C} \) and \( \sum_\ell |\alpha_\ell|^2 = 1 \). Here, \( \ket{\ell} \) represents the basis state associated with the binary encoding of the integer \( \ell \), i.e., \( \ell = \sum_k b_k 2^{n-k} \).

We require the transition function to be local in the sense that the unitary operator can be decomposed into a product of unitaries, each acting non-trivially only on qubits \( k-1 \), \( k \), and \( k+1 \) (modulo \( n \)) for some \( k \). In particular, we focus on a special subclass known as \textit{partitioned unitaries}, defined as:
\begin{equation}
  \left(
    \sum_{k=1}^n \ket{0}_{k+1}\bra{0}_{k} +
    \ket{1}_{k+1}\bra{1}_{k}
  \right) U_1
  \left(
    \sum_{k=1}^n \ket{0}_{k}\bra{0}_{k+1} +
    \ket{1}_{k}\bra{1}_{k+1}
  \right) U_0,
\end{equation}
where \( U_0 = W_0^{\otimes n/2} \), \( U_1 = W_1^{\otimes n/2} \), and \( W_0 \), \( W_1 \) are \( 2 \times 2 \) unitary operators. All indices are taken modulo \( n \), and we assume \( n \) is even. Additionally, we require that \( W_0 \) and \( W_1 \) conserve particle number, meaning $
W_0 \ket{11} = W_1 \ket{11} = \ket{11}, \quad \text{and} \quad W_0 \ket{00} = W_1 \ket{00} = \ket{00}.$

In this work, we formulate the density classification task within the quantum framework of PUQCA and search for unitary transition functions that solve it. As in the classical setting, we restrict our study to one-dimensional lattices and ask: are there unitary operators that solve the DCT after $T$ time steps? Our goal is to evaluate the computational power of PUQCA in this setting and to assess the effectiveness of a genetic algorithm in searching over the continuous space of unitary operations. Unlike the classical case—where the space of rules is finite—the quantum case involves a continuous (and infinite) parameter space, making the problem more challenging.

To fairly evaluate the performance of the PUQCA, we require a classical baseline. Classical CAs with transition functions \( f: \Sigma^{|\mathcal{N}|} \to \Sigma \), where \( \Sigma \) is the finite alphabet and \( \mathcal{N} \) is the neighborhood, are typically not reversible. By contrast, unitary operators are inherently reversible. For this reason, we use \emph{partitioned cellular automata} (PCA)~\cite{ToffoliMargolus,CostaM20}—which conserve particle number and can be made reversible via permutations—as a more suitable classical baseline.

Since PCA also conserves particle number, it cannot rely on classical halting via convergence to homogeneous configurations. Instead, we define the output of the computation through a local measurement at a predetermined site.

Because we employ a one-dimensional lattice for the PUQCA, two tilings are sufficient to generate non-trivial dynamics. The transition function is composed of two layers:
\begin{equation}
\mathcal{E} = \left(\bigotimes_{T_{i}^{(1)} \in \mathcal{T}_1} W_1\right) \left(\bigotimes_{T_{i}^{(0)} \in \mathcal{T}_0} W_0\right), \label{eq:TransitionFunc}
\end{equation}
where each  $W_j \colon \mathbb{C}^2 \otimes \mathbb{C}^2 \rightarrow \mathbb{C}^2 \otimes \mathbb{C}^2$, for \( j \in \{0,1\} \), acts on neighboring cells. We use the shorthand \( \mathcal{E}(W_1, W_0) \) to denote this evolution. The first tiling acts on even sites: $T^{(0)}_i = \{i, i+1\}$ with \( i = 0,2,\ldots,n-2 \), and the second tiling acts on odd sites: \( T^{(1)}_i = \{i, i+1\} \) with \( i = 1,4,\ldots,n-1 \), all modulo $n$, see Figure \ref{fig:PUQCA_rep} for what it would be the quantum circuit for the one-dimensional PUQCA.

\begin{figure}[t]
\centerline{
\Qcircuit @C=1.2em @R=1.2em {
  & \ustick{\vdots}\qw & \multigate{1}{W_0} & \gate{W_1}   &  \qw   &\cdots & &\multigate{1}{W_0} & \gate{W_1}   & \qw\\
  &\qw & \ghost{W_0}        & \multigate{1}{W_1} & \qw  &  \cdots&&\ghost{W_0}        & \multigate{1}{W_1} & \qw \\
  & \qw                 & \multigate{1}{W_0} & \ghost{W_1}        & \qw   &\cdots &&\multigate{1}{W_0} & \ghost{W_1}        & \qw \\
  & \qw & \ghost{W_0}        & \multigate{1}{W_1} & \qw &\cdots &&\ghost{W_0}        & \multigate{1}{W_1} & \qw  \\
  & \qw                 & \multigate{1}{W_0} & \ghost{W_1}        & \qw  &\cdots && \multigate{1}{W_0} & \ghost{W_1}        & \qw \\
  & \qw & \ghost{W_0}        & \multigate{1}{W_1} & \qw   &\cdots&&\ghost{W_0}        & \multigate{1}{W_1} & \qw \\
  & \qw                 & \multigate{1}{W_0} & \ghost{W_1}        & \qw  &\cdots & &\multigate{1}{W_0} & \ghost{W_1}        & \qw\\
  & \qw & \ghost{W_0}        & \multigate{1}{W_1} & \qw  &  \cdots& &\ghost{W_0}        & \multigate{1}{W_1} & \qw \\
  & \dstick{\vdots}\qw  & \gate{W_0}             & \ghost{W_1}        & \qw  &\cdots &&\gate{W_0}             & \ghost{W_1}        & \qw 
  \gategroup{1}{3}{9}{4}{.8em}{--}\\
  && &  \mathcal{E}(W_1,W_0)&&&&&&
}}
\caption{\label{fig:PUQCA_rep}Quantum-circuit representation of a one-dimensional PUQCA. Each time step is \(\mathcal{E}(W_{1}, W_{0})\): first apply the two-qubit gate \(W_{0}\) on even bonds, then \(W_{1}\) on odd bonds. Layers are drawn as continuing above and below to indicate repetition, and periodic boundary conditions connect the first and last qubits.}
\end{figure}
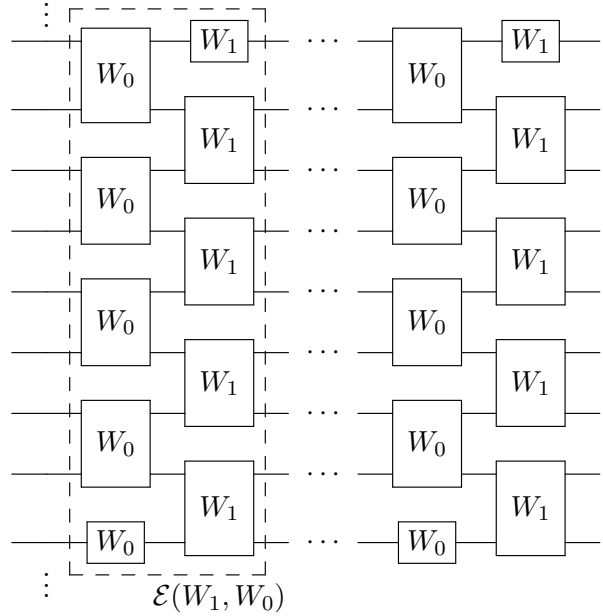

Though each \( W_j \in \mathbb{C}^{4 \times 4} \), particle number conservation implies it acts as identity on \( \ket{00} \) and \( \ket{11} \), and non-trivially only on the subspace spanned by \( \ket{01} \) and \( \ket{10} \). Thus, it effectively takes the form:
\begin{equation}
\label{eq:U2}
\begin{bmatrix}
1 & 0 & 0 & 0 \\
0 & \ast & \ast & 0 \\
0 & \ast & \ast & 0 \\
0 & 0 & 0 & e^{i\phi}
\end{bmatrix},
\end{equation}
where the middle block corresponds to a general \( U(2) \) matrix. The general form of such a \( W_j \) operator can be parameterized by four angles \( (\theta, \alpha, \gamma, \xi) \), as:
\begin{equation}
\label{eq:parametri}
W_j(\theta,\alpha,\gamma,\xi,\phi) =
\begin{bmatrix}
1 & 0 & 0 & 0 \\
0 & e^{i\xi} \cos\theta & e^{i\gamma} \sin\theta & 0 \\
0 & -e^{i(\alpha - \gamma)} \sin\theta & e^{i(\alpha - \xi)} \cos\theta & 0 \\
0 & 0 & 0 & e^{i\phi}
\end{bmatrix},
\end{equation}
where \( j \in \{0,1\} \). For the purposes of our current investigation, we fix \( \phi = 0 \), thereby removing any interaction between the states \( \ket{00} \) and \( \ket{11} \), and simplifying the search for effective transition functions \( \mathcal{E} \).

\subsection{The density classification problem}
\label{sec:background/density_classification}

After the conceptual explanation of density classification, we provide a more rigorous definition in this section to solve the DCT with QCAs. Subsequently, we introduce a method to verify the performance of the transition functions in solving the task.

We start with a few ancillary definitions before providing the main one.

\begin{defn}
Let $n$ be a natural number. The \emph{majority function}
$\operatorname{maj}: \{0, 1\}^n \rightarrow \{-1,0, + 1\}$ is defined so that
\begin{equation}
    \operatorname{maj} \left( b_0, \ldots, b_{n-1} \right) =\begin{cases}
 +1\; , \; \sum_{i=1}^n b_i > n/2 \; \\
 -1\; ,\; \sum_{i=1}^n b_i < n/2 \; \\
 0\; ,\;  \text{Otherwise.} \;
 \end{cases}
\end{equation}
\end{defn}

\begin{defn}
Let $g: \{0, 1\}^n \rightarrow \{-1,0,+1\}$ be a random-valued function.
Then $g$ is said to accomplish the \emph{density classification task}
for $\delta \geq 0$ if for every $n$-bit string $\vb{b} = (b_0, \ldots, b_{n-1})$ with
$\sum_i b_i \neq \frac{n}{2}$ we have for $\operatorname{maj}(\vb{b})=1$,
$\Pr(g(\vb{b}) = \operatorname{maj}(\vb{b})) \geq \frac{1}{2} + \delta$ and $\Pr(g(\vb{b}) = \operatorname{maj}(\vb{b})) \leq \frac{1}{2} + \delta$ for $\operatorname{maj}(\vb{b})=-1$ with $\delta \neq 0$. In the case that we have $\delta=0$ the previous inequalities become strict. We refer to the function $g$ as the \emph{guessing} function.
\end{defn}

\begin{defn} \label{defn:density_classifier}
Let $\mathcal{A}$ be an $n$-cell PUQCA. We say that $\mathcal{A}$ is a
\emph{density classifier} for location $p$ at time $t$ if the following
procedure computes a function that solves the density classification task.
\begin{enumerate}
  \item Encode the $n$-bit input string into an $n$-qubit quantum register.
  \item Apply $\mathcal{A}$ to the quantum register $t$ times.
  \item Measure the $p^\text{th}$ qubit of the quantum register in
  the computational basis.
\end{enumerate}
\end{defn}

Based on the above definition, we see that a PUQCA $\mathcal{A}$
is a density classifier for location $p$ at time $t$ if the
result of the measurement, represented by the
output $g_{p, t} (\vb{b})$ of the random-valued function $g_{p, t}$,
is likely to agree with the value of $\maj(\vb{b})$, that is, the PUQCA $\mathcal{A}$ is a density classifier for location $p$
at time $t$ if $g_{p, t}$ accomplishes the density classification task.

\begin{prob}
\label{prob}
Let $n \in \mathbb{N}$ and $p \in \{0, 1, \ldots, n-1\}$.
Find an $n$-cell PUQCA $\mathcal{A}$ and a natural number $t$, such that
$\mathcal{A}$ is a density classifier for location $p$ at time $t$.
\end{prob}
The key to validating the success of the transition function at a particular time is our definition of the fitness function, presented next.


\section{Approach to the DCT}\label{sec:approach}

\subsection{The fitness function}

Outcomes in quantum physics are probabilistic in nature.  Our  goal here is thus to propose an evaluation criterion that correlates the probabilities of measurements with the density in the initial state. 

The idea is simple. After we have mapped an $n$-bitstring $b_{n-1}\cdots b_1 b_0$ to an $n$-qubit system $\ket{b_{n-1}\cdots b_1 b_0}$,
we choose a location $p$ where the measurement will be performed.

We define the \emph{guessing} function $g_{p, t}:\{0,1\}^n\rightarrow \{0,1\}$ to be the result of the following three steps of \Cref{defn:density_classifier}: 
 \begin{equation}
 \label{eq:guess}
 g_{p,t}(b)=\begin{cases}
 +1\; , \; \left|\bra{1}_{ p} \mathcal{E}^t \ket{b}\right|^2\geq 1/2+\delta \; \\
 -1\; ,\;  \left|\bra{1}_{ p} \mathcal{E}^t \ket{b}\right|^2\leq 1/2-\delta \; \\
 0\; ,\;  \text{Otherwise.} \;
 \end{cases}
 \end{equation}
where the case with $\delta=0$ the inequalities above become strict.

Given that  now we have a criterion to guess the density of a given string, we  define the \emph{fitness function} over a set $\mathcal{B}\subseteq \{\{0,1\}^n |\sum_ib_i \neq n/2\}$ of strings, $F_\delta:\mathcal{B}\rightarrow \{0,1\}$. To do that we rely on the normalised mean absolute error, as follows:
\begin{equation}
\label{eq:Fit}
F_\delta(\mathcal{B})=1-\frac{1}{2|\mathcal{B}|}\sum_{b\in\mathcal{B}}\left|g_\delta(b)-\maj(b)\right|,
\end{equation}
where $|\mathcal{B}|$ is number of strings in $\mathcal{B}$. If we correctly guess the densities of all strings in $\mathcal{B}$, then $F_\delta(\mathcal{B})=1$. If a wrong guess is made for all strings in $\mathcal{B}$, then $F_\delta(\mathcal{B})=0$.

\subsubsection{A Lower bound for $F$}
\label{sec:special}

In this subsection, we derive a lower bound for the fitness function defined in \Cref{eq:Fit}. This bound corresponds to a special case in which, for all \( t \neq 0 \), the transition function does not create any superposition and acts purely as a permutation operator. As we will show, under these circumstances, the fitness calculation for \( t \neq 0 \) is effectively equivalent to that of \( t = 0 \). This case is analogous to employing the  PCA, as described in~\cite{CostaM20}, and therefore serves as a reference for contrasting the classical and quantum models. For clarity, we denote this classical fitness as \( F_C \), with the subscript \( C \) indicating "classical."

In such simple instances, we can exactly count the number of cases in which the measurement outcome leads to incorrect information about the global configuration. Our starting point is a reformulation of \Cref{eq:Fit} as:
\begin{equation}
\label{eq:fc}
F_C = 1 - \frac{\mathcal{W}}{T},
\end{equation}
where \( \mathcal{W} \) is the number of incorrect classification outcomes over a total of \( T \) initial states. 

To illustrate the computation of \( \mathcal{W} \), we begin with simple examples and generalize the result.

\paragraph{Zero excitations.} If all bits are zero, the measurement will always yield the correct classification (majority of 0s), so \( W = 0 \).

\paragraph{One excitation.} For configurations with a single 1 (i.e., \( b_i = 1 \) for some \( i \)), there are exactly \( n \) such configurations. If the measurement is made at position \( p \in \{0, \dots, n-1\} \), the only configuration that results in an incorrect outcome is the one where the excitation is located at position \( i = p \). Therefore, there is exactly one wrong classification, \( \mathcal{W} = 1 \).

\paragraph{Two excitations.} For states with two 1s (i.e., \( b_i = b_j = 1 \), \( i \neq j \)), there are \( \binom{n}{2} \) such configurations. The wrong outcomes occur when one of the excitations is located at the measurement position \( p \), and the other is elsewhere. Thus, there are \( n - 1 \) such cases. Due to the symmetry of the setup, choosing either excitation to be at \( p \) leads to the same count of errors, confirming the total as \( W = n - 1 \).

\paragraph{Three excitations.} For states with three 1s (i.e., \( b_i = b_j = b_\ell = 1 \), \( i \neq j \neq \ell \)), the number of incorrect outcomes corresponds to the number of configurations where one excitation is at the measurement site \( p \), and the remaining two are placed elsewhere. The count is thus \( \binom{n - 1}{2} \) erroneous configurations.

From these patterns, the general formula for the number of wrong classifications \( W \) for the entire space of configurations is:
\begin{align}
\mathcal{W} &= 2 \sum_{i = 1}^{n/2 - 1} \binom{n - 1}{i - 1}\nonumber\\
&=2^{n-1}-2\binom{n - 1}{n/2 - 1},
\end{align}
where the factor of 2 accounts for the symmetry in classification errors—i.e., configurations with more 0s than 1s contribute analogously once the number of 1s exceeds \( n/2 \).

Substituting into \Cref{eq:fc}, we obtain the lower bound for the fitness function for the quantum classifier:
\begin{equation}
\label{eq:Fc}
F \geq F_C.
\end{equation}

The table below presents the lower bounds given by \Cref{eq:Fc} for the fitness function defined in \Cref{eq:Fit}, evaluated for different system sizes of the quantum classifier.

\begin{table}[H]
\centering
\begin{tabular}{c|c}
\toprule
\textbf{System Size \( n \)} & \textbf{Lower Bound \( F_C \)} \\
\midrule
4  & 0.8000 \\
6  & 0.7273 \\
8  & 0.6882 \\
10 & 0.6632 \\
12 & 0.6456 \\
14 & 0.6325 \\
16 & 0.6222 \\
\bottomrule
\end{tabular}
\caption{Lower bounds for the fitness values \( F \) for different values of \( n \)\label{tab:Lowerb}}
\end{table}

\subsection{The evolutionary search algorithm for the DCT}

As introduced earlier, we employ an evolutionary search strategy—implemented as a \textit{genetic algorithm} — to find effective PUQCA-based classifiers for the density classification task (DCT). Once the parameters of the system are specified—namely, the number of qubits \( n \), the number of time steps \( t \), the value of \( \delta \) in \Cref{eq:guess}, and the measurement location \( p \)—the GA searches for configurations that maximize the fitness function \( F_{\delta}(\mathcal{B}) \in [0,1] \). According to the transition function \( \mathcal{E} \) defined in \Cref{eq:TransitionFunc}, the search space is characterized by eight parameters: four angles for each unitary \( W(\theta, \alpha, \gamma, \xi) \).

In the language of evolutionary computation, the parameters under optimization—the angle values—are referred to as \emph{genes}. A complete set of these eight genes constitutes a \emph{chromosome} or \emph{individual}, representing a candidate solution. A collection of individuals forms the \emph{population} \( P_t \), which evolves over \( T \) generations indexed by the time step \( t \). Internally, the population is represented as a matrix of dimensions \texttt{individuals} \(\times\) \texttt{genes}.

For our purposes, the genetic algorithm makes use of only two standard genetic operators: \textbf{selection} and \textbf{mutation}. We omit the crossover operation (i.e., recombination between individuals) in this implementation.

\begin{itemize}
  \item \textbf{Selection}: We use the standard \textit{roulette wheel} selection mechanism to generate a new pool of individuals \( P'_t \), equal in size to the original population. This fitness-proportional method probabilistically selects individuals from \( P_t \) with likelihood proportional to their fitness values—thus favoring better-performing individuals, which may appear more than once in \( P'_t \). In Algorithm~\ref{alg:evol}, the selection process is described by the function \( S: P_t \rightarrow P'_t \), with \( P'_t \subseteq P_t \).

  \item \textbf{Mutation}: The mutation operation perturbs selected genes within the pool \( P'_t \). Since the genes are real-valued angles, mutation is implemented by adding a perturbation drawn from a Gaussian distribution centered at the current gene value. That is, for the \( i \)-th gene \( g_i \), we have:
  \begin{equation}
    g_i \leftarrow g_i + G(\mu, \sigma^2),
  \end{equation}
  where \( G(\mu, \sigma^2) \) is a Gaussian distribution with mean \( \mu = g_i \) and standard deviation \( \sigma \). For each gene in the population, a random number in the interval \([0, 1]\) is generated. If this number is less than the mutation rate \( p_m \), the gene is mutated; otherwise, it remains unchanged. The mutation is thus governed by the parameters \( p_m \) and \( \sigma \), both held fixed during the evolution. The overall mutation process is described by the function \( M_{p_m, \sigma}: P'_t \rightarrow P_t \).
\end{itemize}

\begin{algorithm}[H]
  \caption{Evolutionary Search for Quantum Classifiers}
  \label{alg:evol}
  \begin{algorithmic}[1]
    \State \textbf{Input:} \( (n, \delta, \sigma, p_m, p, T) \)
    \State \( t \leftarrow 0 \)
    \State Initialize a population \( P_t \)
    \State Compute fitness \( F \) for each individual in \( P_t \)
    \While{\( t \leq T \)}
      \State \( P'_t \leftarrow S(P_t) \) \Comment{Selection}
      \State \( P_t \leftarrow M_{p_m, \sigma}(P'_t) \) \Comment{Mutation}
      \State Compute fitness \( F \) for each individual in \( P_t \)
      \If{there exists \( i \in P_t \) such that \( F(P_t(i)) == 1 \)}
        \State \textbf{Break}
      \Else
        \State \( t \leftarrow t + 1 \)
      \EndIf
    \EndWhile
    \State \textbf{Output:} \( P_T(i) \) where \( F(P_T(i)) \geq F(P_T(k)) \) for all \( k \in P_T \)
  \end{algorithmic}
\end{algorithm}

Algorithm~\ref{alg:evol} outlines the evolutionary optimization process. The population is first initialized with individuals—each a tuple of randomly sampled parameters \( (\theta_1, \alpha_1, \gamma_1, \xi_1, \theta_2, \alpha_2, \gamma_2, \xi_2) \in [0, 2\pi]^8 \). Fitness is then evaluated using the function \( F \), and the population evolves over \( T \) generations. During each iteration, selection and mutation are applied sequentially. The algorithm halts early if an individual achieves perfect fitness \( F = 1 \), indicating a valid classifier for the specified task parameters.

%


\section{Symmetries and existence of PUQCA - based solutions}\label{sec:transl_invaDct}

\subsection{Shift invariance and measurement Equivalence}

Here, we demonstrate how the translational invariance of the PUQCA can, in principle, be utilised to reduce the computational complexity of the DCT. The evolution operator $\mathcal{E}$ of the PUQCA is defined as the product of two unitary operators, each of which is invariant under the transformation
\begin{equation}
    T^2 \ket{b_0}\ket{b_1}\ket{b_2} \ldots \ket{b_{n-3}}\ket{b_{n-2}}\ket{b_{n-1}} 
    := \ket{b_2}\ket{b_3}\ket{b_4} \ldots \ket{b_{n-1}}\ket{b_0}\ket{b_1},
\end{equation}
where \( b_0, b_1, \ldots \) are binary values. This operation corresponds to a two-step cyclic shift, and we note that our PUQCA commutes with this transformation:
\begin{equation}
    [\mathcal{E}, T^2] = 0,
\end{equation}
where the operator \( T \) is defined as the one-step cyclic shift:
\begin{equation}
    T \ket{b_0}\ket{b_1}\ket{b_2} \ldots \ket{b_{n-3}}\ket{b_{n-2}}\ket{b_{n-1}} 
    := \ket{b_1}\ket{b_2}\ket{b_3} \ldots \ket{b_{n-2}}\ket{b_{n-1}}\ket{b_0}.
\end{equation}

When we have a PUQCA on a one-dimensional ring of \( n \) qubits (i.e., the \( n \)-cycle), this translational symmetry implies that the evolution operator satisfies
\begin{equation}
\label{eq:TranslInv}
    [\mathcal{E}, T^{2m}] = 0,
\end{equation}
for all \( m \in \mathbb{N} \), with \( 2m \bmod n \). This symmetry property has implications for solving the density classification task, as we demonstrate in the following two lemmas.

\begin{lemma}[Translation covariance of PUQCA-based density classification]\label{lemma:transl} An $n$-cell PUQCA $\mathcal{A}$ which is a density classifier for a $n$-qubit input system $\ket{b}$ for a location $p$ at time $t$ is also an $n$-cell PUQCA $\mathcal{A}$ for a $n$-qubit input system $\ket{b'}$, such that $\ket{b'}= T^{2m}\ket{b}$, for a location $p\oplus 2m$ at time $t$. 
\end{lemma}
\begin{proof}
By saying that $\mathcal{A}$ is a density classifier for an $n$-qubit system for $\ket{b}$ is equivalent to say that the probability
\begin{equation}
    Pr(1_p) = \left|\bra{1}_{ p} \mathcal{E}^t \ket{b}\right|^2,
\end{equation}
solves \Cref{prob} for the input state $\ket{b}$. Now, from the translational invariance of $\mathcal{A}$~\Cref{eq:TranslInv}, we have
\begin{align}
    Pr(1_p) &= \left|\bra{1}_{ p} \mathcal{E}^t \ket{b}\right|^2\nonumber\\
    &=\left|\bra{1}_{ p} (T^{2m})^{\dagger}T^{2m}\mathcal{E}^t(T^{2m})^{\dagger}T^{2m} \ket{b}\right|^2\nonumber\\
    &=\left|\bra{1}_{ p'} T^{2m}\mathcal{E}^t(T^{2m})^{\dagger} \ket{b'}\right|^2\nonumber\\
    &=\left|\bra{1}_{ p'} \mathcal{E}^t \ket{b'}\right|^2\nonumber\\
    &= Pr(1'_{p'}),
\end{align}
where $p'=p \oplus 2m$ for $2m\mod{n}$.
\end{proof}

\subsection{Existence of solution}
\label{sub:equaSup}

The aim of this section is to proof that for a given $\delta \neq 0$ the states in quantum mechanics known as \textit{Dicke} states~\cite{Dicke}  
\begin{equation}
\label{eq:Dic}
\ket{D_n^{i}}=\frac{1}{\sqrt{C_n^{i}}}\sum_{j}\mathcal{P}_j\left[\ket{1}^{\otimes i}\otimes \ket{0}^{\otimes (n-i)}\right],
\end{equation}
where $\sum_{j}\mathcal{P}_j\left[\cdot\right]$ denotes the sum over all possible permutations and $C_n^{i}=\binom{n}{i}$ is the binomial coefficient, provide a solution for the classification problem.

\begin{theorem}[Dicke state solutions to the DCT]\label{theo:Dicke}
The Dicke states $\ket{D^{i}_n}$, for  $i\neq n/2$ resolve the classification problem for the $n$-cell PUQCA, where $n$ is an even number and $n>2$.
\end{theorem}
\begin{proof}
Consider that we have an $n$-cell PUQCA with $n>2$ and the PUQCA state is given by
\begin{equation}
\label{eq:Dic2}
\ket{D_n^{i}}=\sqrt{\frac{i!}{n(n-1)\cdots (n-i+1)}}\sum_{j}\mathcal{P}_j\left[\ket{1}^{\otimes i}\otimes \ket{0}^{\otimes (n-i)}\right],
\end{equation}
where $i<n/2$. Now we compute the probability to measure one excitation in the cell $p$. In order to make this calculation we have to consider all states where there is one excitation located in the cell $p$. These states can be easily estimated observing that for an excitation located in the cell $p$ there are $n-1$ vacancy spaces for the $i-1$ remainders excitations. From this observation we see that there are  $C_{n-1}^{i-1}=\binom{n-1}{i-1}$ states which provide the right answer for $\operatorname{maj}(b)=-1$. Thus, the probability to measures such states is given by
\begin{equation}
    \left|\braket{1_p}{D_n^i}\right|^2=\frac{i!}{n(n-1)\cdots (n-i+1)}\frac{(n-1)\cdots(n-i+1)}{(i-1)!}=\frac{i}{n}.
\end{equation}
Since we are considering the case where $i<n/2$, we conclude that
\begin{equation}
    \norm{\braket{1_p}{D_n^i}}^2<\frac{1}{2}.
\end{equation}
For the cases where $i>n/2$ the prove is completely analogue. But here the correct answers are associated now to $\operatorname{maj}(b)=+1$.
\end{proof}

From \Cref{theo:Dicke}, we conclude that if there exists a transition function \( \mathcal{E}^{(D_n)} \) capable of generating the Dicke state \( \ket{D_n^i} \) from any \( n \)-qubit input state with \( i \) excitations—i.e.,
$\mathcal{E}^{(D_n)t} \ket{b_n^i} = \ket{D_n^i}$,
for all \( i \in \{1, \ldots, \tfrac{n}{2} - 1, \tfrac{n}{2} + 1, \ldots, n - 1\} \), after a fixed number of time steps \( t \), then the corresponding PUQCA constitutes a valid density classifier. Notably, in this scenario, the classification is independent of the specific measurement location on the lattice. However, it remains unclear whether such Dicke states can be generated solely through local unitary operations. Nonetheless, the theoretical existence of these ideal solutions offers valuable guidance for designing transition functions that, at minimum, approximate the evolution toward the corresponding Dicke state for a given initial excitation number.

An important observation is that the relative phases among the basis states of a Dicke state do not affect the marginal distribution of any single qubit. This property enables the construction of a unitary operator that allows density classification to succeed by measuring a single qubit. 

We begin by considering unitaries that conserve the total excitation number, i.e., those that commute with the total spin-$z$ operator:
\begin{equation}
[U, J^z] = 0, \quad \text{where} \quad J^z = \frac{1}{2} \sum_{k=1}^{n} \sigma^z_k.
\end{equation}
This condition ensures that the output state preserves the Hamming weight (or density) of the input configuration.

Let \( \ket{J, M} \) denote the eigenstates of \( J^z \), where \( J = n/2 \) and \( M \in \{-n/2, -n/2 + 1, \dots, n/2\} \). Any unitary \( U \) satisfying the commutation relation above can be expressed in block-diagonal form as
\begin{equation}
U = \bigoplus_{i=0}^{n} U^{(i)},
\end{equation}
where the block \( U^{(i)} \) acts on the subspace with exactly \( i \) excitations (i.e., Hamming weight \( i \)) and has dimension \( \binom{n}{i} \).

To ensure density classification, each block \( U^{(i)} \) must map the classical basis states with \( i \) ones to an orthonormal set of quantum states. A suitable choice is the set of generalized Dicke states \( \{ \ket{D_n^i, k} \} \), which are Dicke states modified by relative phases across the computational basis. Specifically, we define:
\begin{equation}
\ket{D_n^i, k} := Z^k(n, i) \ket{D_n^i},
\end{equation}
where
\begin{equation}
Z^k(n, i) = \sum_{j=0}^{\binom{n}{i} - 1} e^{i \frac{2\pi k j}{\binom{n}{i}}} \ket{j}\bra{j},
\end{equation}
and \( \ket{j} = \mathcal{P}_j \left[ \ket{1}^{\otimes i} \otimes \ket{0}^{\otimes (n - i)} \right] \), i.e., the \( j \)-th distinct permutation of a state with \( i \) ones followed by \( n - i \) zeros. By construction, the states \( \ket{D_n^i, k} \) form an orthonormal basis:
\begin{equation}
\langle D_n^{i}, k|D_n^{i'}, k'\rangle = \delta_{i, i'} \delta_{k, k'}.
\end{equation}

We can now define the block \( U^{(i)} \) as:
\begin{equation}
U^{(i)} = \sum_{k = 0}^{\binom{n}{i} - 1} \ket{D_n^i, k} \bra{k},
\end{equation}
where \( \ket{k} = \mathcal{P}_k \left[ \ket{1}^{\otimes i} \otimes \ket{0}^{\otimes (n - i)} \right] \) is the \( k \)-th classical basis state with \( i \) excitations. This construction defines a unitary \( U \) that transforms each excitation subspace to a corresponding orthonormal set of generalized Dicke states.

This choice of \( U \) allows density classification to be performed by measuring any single qubit, as all marginal qubit distributions are consistent across the states in each block. Moreover, any post-processing unitary \( W \) that acts trivially on the measurement qubit (say qubit \( s \))—i.e., \( W = W' \otimes I_s \)—will preserve the reduced density matrix \( \rho_s \), and hence also suffice for classification. This flexibility highlights the robustness of this approach under local post-processing.


\section{Results I}
\label{sec:res1}

We begin by presenting the list of solutions in \Cref{tb:res1} and \Cref{tb:res2}, corresponding to system sizes ranging from 4 to 14 cells (qubits). In all cases, the optimisation process was carried out using the parameters \(p_m = 0.36\), \( \sigma = 0.45\), and \( \delta = 0 \), with a population size varying between 100 and 180, and a maximum of 100 generations. The measurement point was fixed at $ p = 1$, with cell indices labelled from 0. Additionally, for each instance, the angles defining the transition functions were optimised for a time \( t \) equal to half the number of cells in the system, where each step we have a sequence of two operators.

\begin{table}[h!]
\centering
 \begin{tabular}{||c c c ||} 
 \hline
 4 Cells & 6 Cells & 8 Cells \\ [0.5ex] 
 \hline\hline
 $\theta_1= 0.6892$, $\xi_1= 6.0897
$ & $\theta_1 = 0.9001$, $\xi_1 = 0.0752$ & $\theta_1= 0.8870$, $\xi_1 = 5.0322$  \\ 
 $\gamma_1=0.1366$, $\alpha_1=0.4177$~ & $\gamma_1 = 0.3960$, $\alpha_1 = 0.4801$~ & $\gamma_1 = 1.6581$, $\alpha_1 = 0.1687$~  \\
   &  &   \\
 $\theta_2=0.8995$, $\xi_2=5.8840$ & $\theta_2 = 0.6786$, $\xi_2 = 0.9749$ & $\theta_2 = 0.5313$, $\xi_2 = 6.2298$  \\
 $\gamma_2=0.8684$, $\alpha_2=2.8956$~ &  $\gamma_2 = 0.5549$, $\alpha_2 = 0.7985    $~ & $\gamma_2 = 1.1563$, $\alpha_2 = 1.6556$ \\ [1ex] 
 \hline
 \end{tabular}
 \caption{List of angles that solve the classification problem from 4 to 8 Cells. \label{tb:res1}}
\end{table}

In \Cref{Fig:sol8}, we illustrate the solution from \Cref{tb:res1} for the case with 8 cells, applied to two different input states containing two excitations. These two states are related by a translation, as described in \Cref{lemma:transl}. As a result, the probability distributions for measuring an excitation in each cell are expected to be shifted accordingly. In this particular example, since \( \ket{b'} = T^2 \ket{b} \), the probability distribution is shifted by two cells.

\begin{figure}[H]
	\includegraphics[trim=-400 0 130 0,clip,scale=0.16]{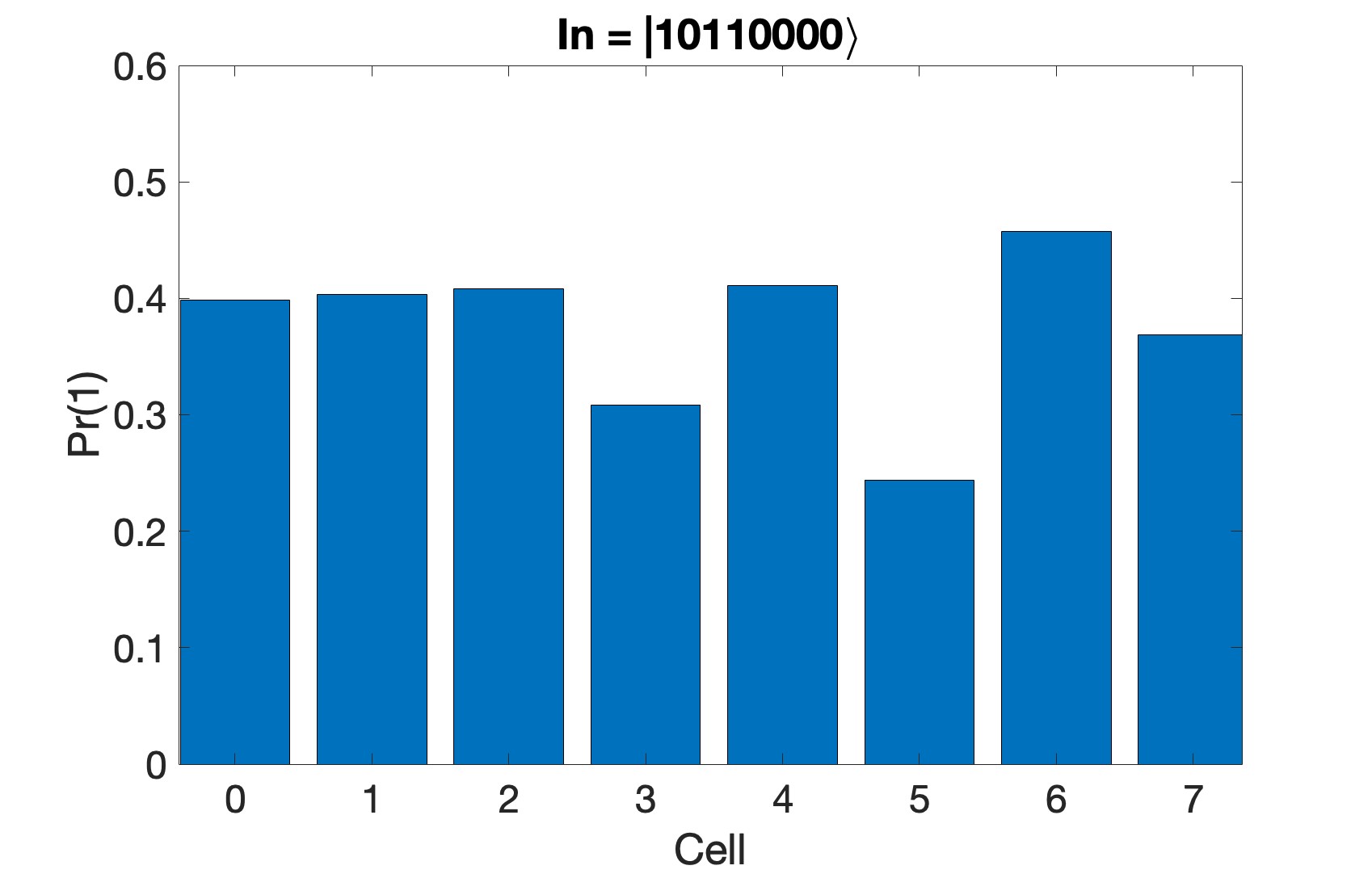}(a)\\
	\includegraphics[trim=-400 0 0 -80,clip,scale=0.16]{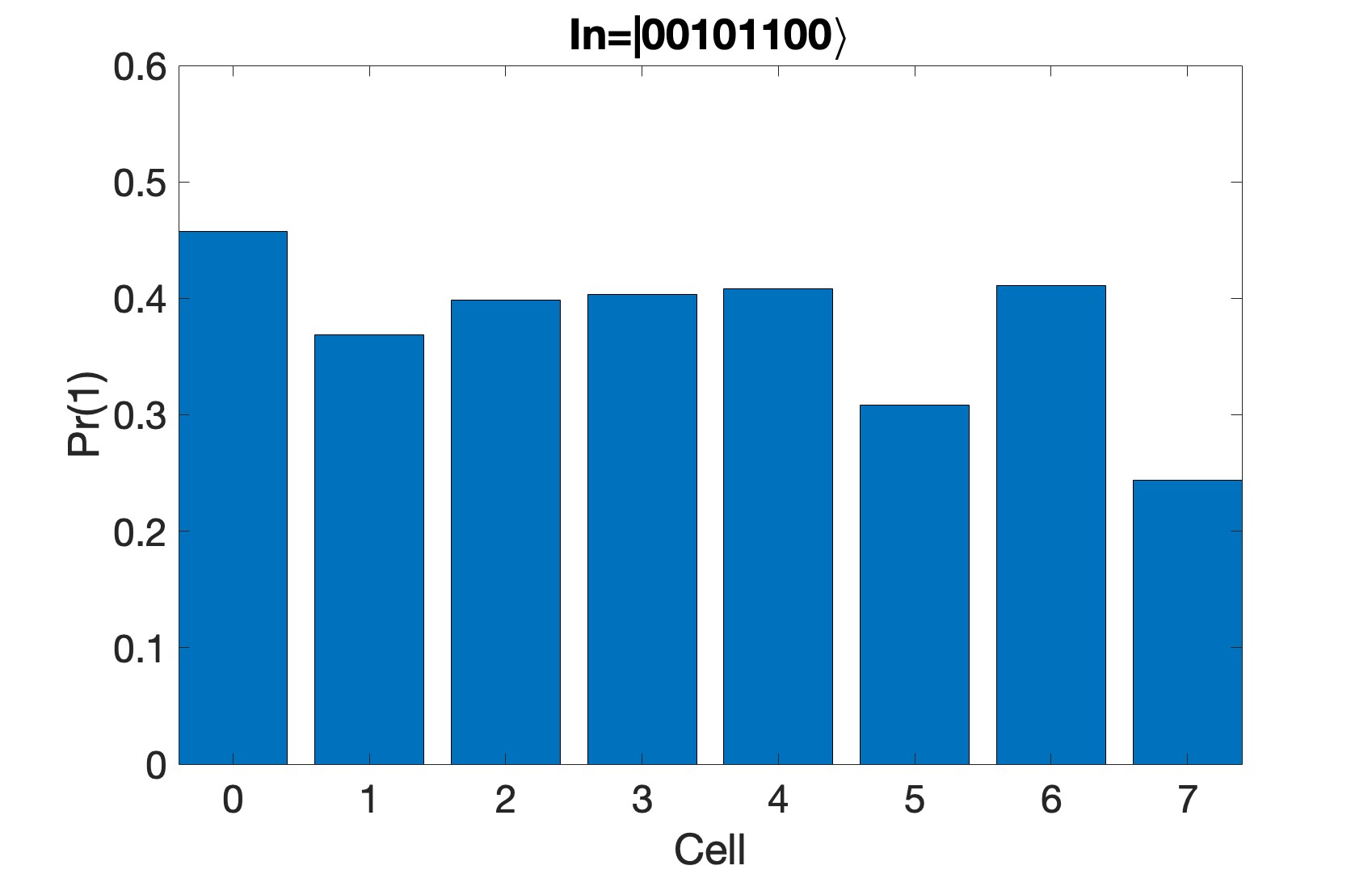}(b)\\
	\caption{\label{Fig:sol8}  These plots show the result of the quantum classifier solution for 8 cells, see \Cref{tb:res1}, with 3 excitations after 4 time steps. The $y$-axis represents the probability of measuring an excitation at each cell, indexed on the $x$-axis. In (a), the initial state is \( \ket{b} = \ket{10110000} \), while in (b), the state is \( \ket{b'} = T^2 \ket{b} = \ket{00101100} \), which is related to \( \ket{b} \) by a two-step cyclic translation.}
\end{figure}

\begin{table}[h!]
\centering
 \begin{tabular}{||c c c ||} 
 \hline
 10 Cells & 12 Cells & 14 Cells \\ [0.5ex] 
 \hline\hline
 $\theta_1= 1.0642$, $\xi_1= 0.3063
$ & $\theta_1 = 0.7497$, $\xi_1 = 0.6441$ & $\theta_1= 0.9282$, $\xi_1 = 5.6197$  \\ 
 $\gamma_1=0.2655$, $\alpha_1=0.8413$~ & $\gamma_1 =5.0980$, $\alpha_1 = 1.7034$~ & $\gamma_1 = 0.4424$, $\alpha_1 = 1.0604$~  \\
   &  &   \\
 $\theta_2=0.8338$, $\xi_2=0.8077$ & $\theta_2 = 0.9875$, $\xi_2 = 0.7946$ & $\theta_2 = 0.6703$, $\xi_2 = 0.6944$  \\
 $\gamma_2=6.1987$, $\alpha_2=2.6387$~ &  $\gamma_2 = 0.3631$, $\alpha_2 = 1.4459    $~ & $\gamma_2 = 3.1894$, $\alpha_2 = 1.8323$ \\ [1ex] 
 \hline
 \end{tabular}
 \caption{List of angles that solve the classification problem from 10 to 14 Cells. \label{tb:res2}}
\end{table}

A more compelling result is presented in \Cref{tab:fitmult}, where we report a single rule that successfully solves the DCT for multiple system sizes. This solution was obtained by designing a cost function that simultaneously optimized performance over systems of 4, 6, and 8,  cells, with time steps \( t = 2, 3, 4 \), respectively. As shown in \Cref{tab:fitmult}, the resulting rule achieves perfect classification for all initial states in the 4-, 6-, and 8-cell systems, and attains a fitness of 0.9909 for the 10-cell system. That is, out of 772 initial configurations with non-half density, only seven instances fail to be correctly classified.

\begin{figure}[H]
	\includegraphics[trim=-400 0 0 -80,clip,scale=0.22]{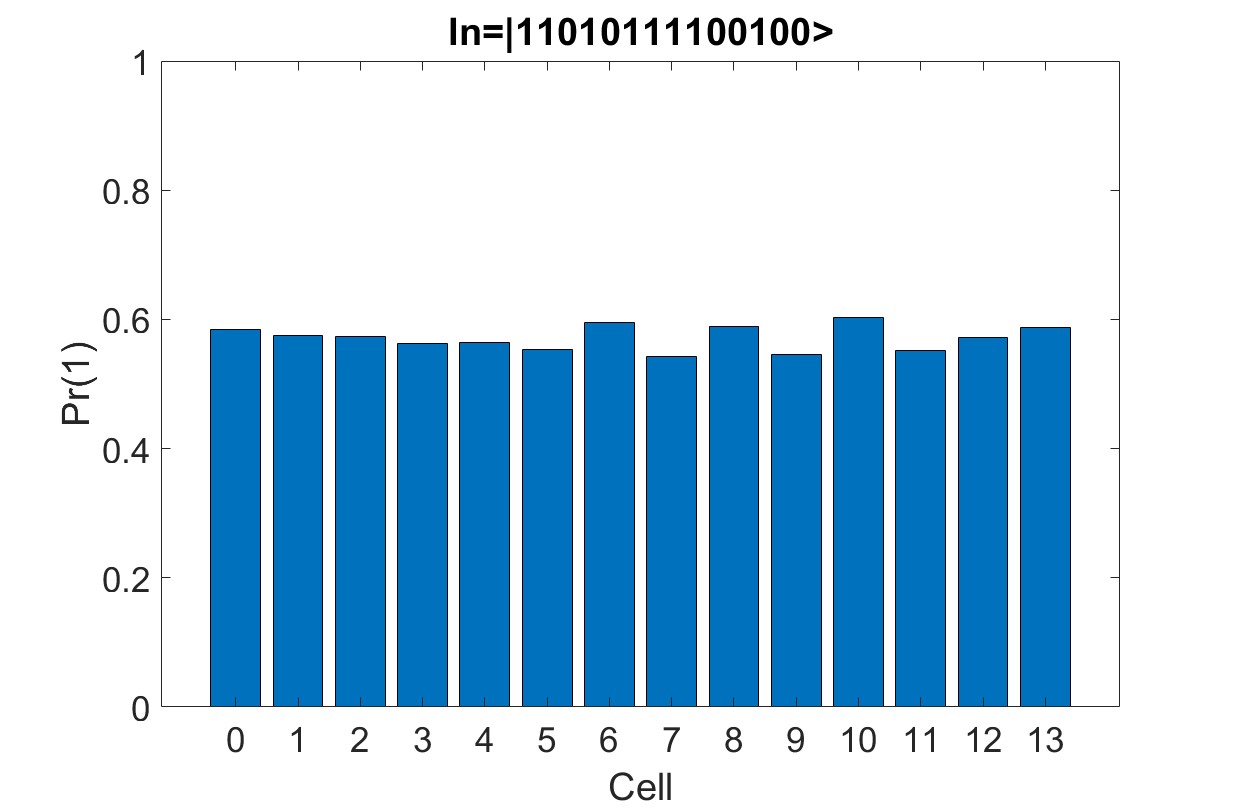}(a)\\
	\includegraphics[trim=-400 0 0 -80,clip,scale=0.22]{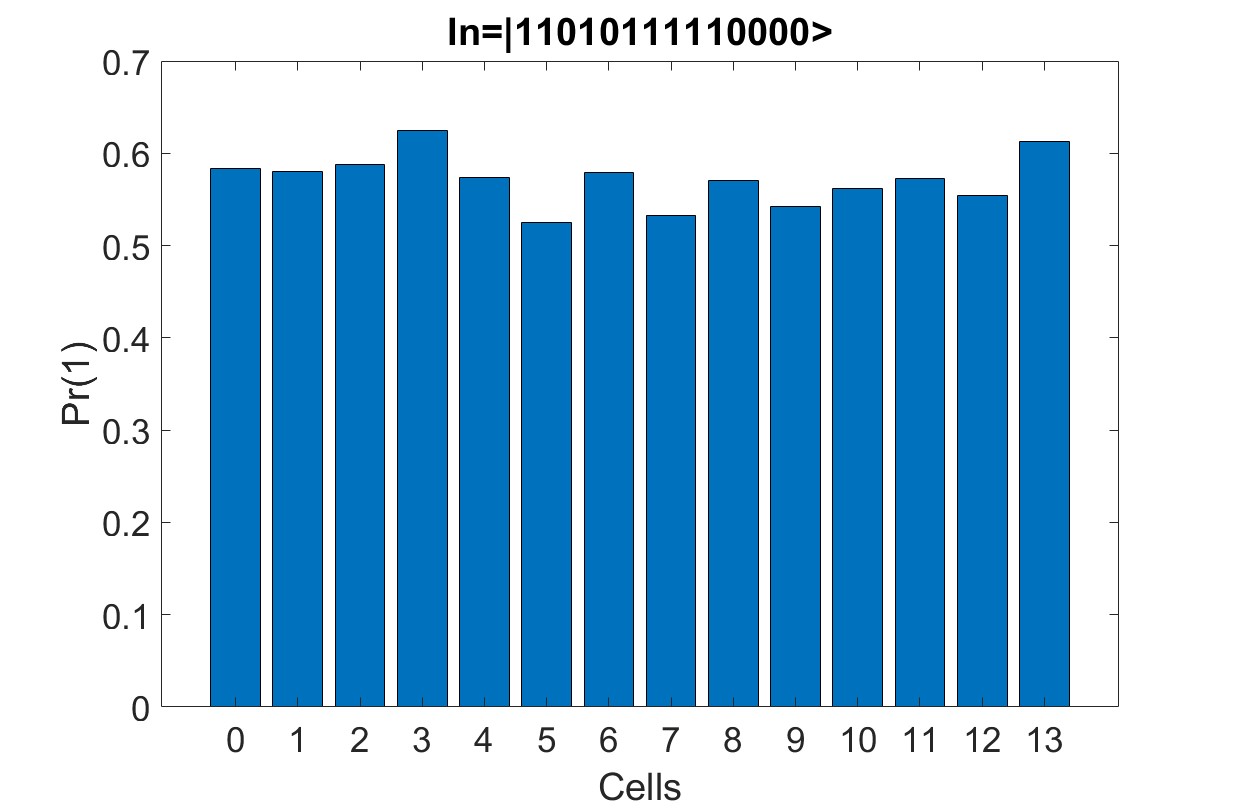}(b)\\
	\caption{
		{\label{Fig:sol12} Here are the plots for the quantum classifier with the rule given in Table \ref{tab:fitmult} with 14 cells and 8 particles after 7 time steps. Likewise, in the case of 8 qubits, the values in the bars are the probabilities to measure one excitation in a given cell, $x$-axis. In (a) our initial state is given by $\ket{b}=\ket{11010111100100}$ and in (b) we have the state $\ket{b'}=\ket{11010111110000}$. }
	}
\end{figure}

Interestingly, this rule—although optimized only for those specific sizes—generalizes well to larger system sizes not included in the optimization. For example, in the case of \( n = 12 \) and \( t = 6 \), it correctly classifies 3128 out of 3172 valid initial states. For \( n = 14 \) and \( t = 7 \), the rule solves 12,850 out of 12,952 cases. Due to the increasing computational cost, we did not extend the testing to even larger systems.

These results suggest the promising possibility that a single rule might exist that generalises to all system sizes, or that optimising over a sufficiently broad set of sizes simultaneously could yield rules with strong generalisation to unseen dimensions.

\begin{table}[H]
\centering
\begin{tabular}{c|c}
\toprule
\textbf{System Size \( n \)} & \textbf{Fitness Value \( F \)} \\
\midrule
4  & 1.0000 \\
6  & 1.0000 \\
8  & 1.0000 \\
10 & 0.9909 \\
12 & 0.9861 \\
14 & 0.9921 \\
\bottomrule
\end{tabular}
\caption{
Fitness values \( F \) achieved by a single PUQCA rule across multiple system sizes, where $n$ gives the number of cells. The rule is defined by the parameters: 
\( \theta_1 = 0.5367 \), \( \xi_1 = 6.1344 \), \( \gamma_1 = 0.2414 \), \( \alpha_1 = 1.0714 \), 
\( \theta_2 = 0.5988 \), \( \xi_2 = 6.1795 \), \( \gamma_2 = 5.0084\), \( \alpha_2 =  1.0913\). 
\label{tab:fitmult}}
\end{table}

We conclude this set of results by examining the special case $\gamma_{1}=\gamma_{2}=0$, which places the PUQCA in an efficiently classically simulable regime, as detailed in the next section. In all tested instances, we found a single rule that performs well across multiple system sizes. However, obtaining a \emph{single} rule that solves the task \emph{uniformly} across system sizes remains substantially more challenging in this setting.

Within this restricted regime, we present a representative rule. Under these constraints, perfect classification is achieved only for $n=4$, while the highest classification efficiency among larger sizes is observed at $n=10$. Although overall performance is markedly inferior to that of the full PUQCA, this rule remains a promising candidate for size-robust behaviour in a setting that admits efficient classical simulation. In the next section, we establish conditions under which the PUQCA is efficiently classically simulable and subsequently evaluate its performance on the DCT in that regime.

\begin{table}[H]
\centering
\begin{tabular}{c|c}
\toprule
\textbf{System size $n$} & \textbf{Fitness $F$} \\
\midrule
4  & 1.0000 \\
6  & 0.9091 \\
8  & 0.8925 \\
10 & 0.9974 \\
12 & 0.8153 \\
14 & 0.7627 \\
\bottomrule
\end{tabular}
\caption{
Single PUQCA rule (with $\alpha_{1}=\alpha_{2}=0$) applied across sizes; perfect performance only for $n=4$ and a large accuracy for $n=10$. Parameters:
$\theta_1 = 0.4749$, $\xi_1 = 1.4803$, $\gamma_1 = 5.5727$;
$\theta_2 = 0.4833$, $\xi_2 = 1.1740$, $\gamma_2 = 3.3965$.
\label{tab:fitmult_simReg_A}}
\end{table}


\section{On the classical simulability of the PUQCA}
\label{sec:perf}

From the results in \Cref{sec:res1}, we observed the existence of a PUQCA $\mathcal{A}$ that functions as a density classifier across different system sizes. Notably, $\mathcal{A}$ possesses a key structural property in its two-qubit unitary: by setting $\alpha = 0$ in \Cref{eq:U2}, the submatrix acting on the subspace $\{\ket{01}, \ket{10}\}$ transitions from the unitary group $U(2)$ to the special unitary group $SU(2)$. This structural shift enables a mapping of the problem into the fermionic representation, which in turn allows for efficient computation of the output probabilities of $\mathcal{A}$ using classical resources.

The main objective of this section is to demonstrate how the fermionic mapping enables classical simulation of the probabilities associated with $\mathcal{A}$. This framework establishes a foundation for extending the search for density classifiers to larger lattice sizes.

To proceed to the fermionic picture, we first have to look for local Hamiltonians, described in terms of spin operators, where its unitary transformation has the following general format
\begin{equation}
\label{eq:SU2}
\begin{bmatrix}e^{i\phi} &  & 0\\
& SU(2)\\
0 &  & e^{-i\phi}
\end{bmatrix},
\end{equation}
which is the type of unitary used for our QCA when $\alpha, \phi=0$. The Lemma below show the parameterised Hamiltonian that gives our targets QCA unitaries.

\begin{lemma}
\label{lemm:Ising_QCA}
The Hamiltonian acting on qubits $k$, $k+1$
\begin{equation}
\label{eq:QCAhamilt}
    h_k = \frac{J}{4}e^{i\theta} \sigma^+_k \sigma^-_{k+1} + \frac{J}{4}e^{-i\theta} \sigma^-_k \sigma^+_{k+1}+h_1(\sigma_z)_k+h_2(\sigma_z)_{k+1},
\end{equation}
for $J,\theta,h_1,h_2,\tau\in \mathbb{R}$
generates the unitary
\begin{equation}
W=e^{-ih_k \tau}=\begin{bmatrix}e^{i\phi} & 0 & 0 & 0\\
0 & \cos{\Omega}-i\cos{\beta} \sin{\Omega} &  -ie^{i\theta}\sin\beta\sin{\Omega}& 0\\
0 & -ie^{-i\theta}\sin\beta\sin{\Omega}& \cos{\Omega}+i\cos{\beta} \sin{\Omega} & 0\\
0 & 0 & 0 & e^{-i\phi}
\end{bmatrix},
\end{equation}
with $\Omega=\tau\sqrt{(h_1-h_2)^2+J^2} ,\ \phi=(h_1+h_2)t,\ \tan\beta=
\frac{J}{h_1-h_2}$.
\end{lemma}
The proof can be contemplated in \Cref{lemmaProof}.

Once we know which local Hamiltonian gives the target QCA unitary operator, we can describe how to get its fermionic correspondent Hamiltonian. For $n$ qubits on a 1D lattice (periodic or not) there is a convenient mapping of the Hilbert space to that of $n$ spinless fermionic modes. Consider an ordering of qubits from $0$ to $n-1$, for which Pauli operators $\sigma^{\alpha}_{n+k}\equiv \sigma^{\alpha}_k$ can be identified in the case of periodic boundaries. For a chain of spin in one dimension, the Jordan Wigner transformation from the spin operator at site $j$, to a Majorana fermion at the same size is defined as
\begin{equation}
\label{eq:mapping}
S^{+}_j =a_j^{\dagger}e^{i\phi_j}\;\text{and}\; S^{-}_j =a_je^{-i\phi_j}
\end{equation}
where $S^{+}=1/2 \sigma^+=1/2\ket{0}\bra{1}$, $S^-=1/2\sigma^-=(\sigma^+)^{\dagger}$ and 
\begin{equation}
\label{eq:stringOp}
\phi_j = \pi \sum_{l=0}^{j-1}a_l^{\dagger}a_l,    
\end{equation}
is a non-local operator. Notice that $(\sigma_z)_j=1-2a^{\dagger}_ja_j$.
Nearest neighbour two qubit interactions like those in \Cref{eq:parametri} map onto nearest neighbour coupling of fermionic modes with the provision that the unitary action on the subspaces $\{\ket{00},\ket{11}\}$ and
$\{\ket{01},\ket{10}\}$ is of type $SU(2)$ (determinant one). For example, the Hamiltonian
acting on qubits $k,k+1$. We can then conclude this part with the lemma that shows the result of mapping our QCA Hamiltonian into the fermionic space.

\begin{lemma}
\label{lem:ferm_halm}
In the fermionic space, the Hamiltonian 
\begin{equation}
  h_k = \frac{J}{4}e^{i\theta} \sigma^+_k \sigma^-_{k+1} + \frac{J}{4}e^{-i\theta} \sigma^-_k \sigma^+_{k+1}+h_1(\sigma_z)_k+h_2(\sigma_z)_{k+1},    
\end{equation}
is given by 
\begin{equation}
\label{eq:FermQCA}
h_k=Je^{i\theta} a_k a^{\dagger}_{k+1}+Je^{-i\theta} a_{k+1} a^{\dagger}_{k}+h_1(1-2a^{\dagger}_ka_k)+h_2(1-2a^{\dagger}_{k+1}a_{k+1}).
\end{equation}
\end{lemma}
The proof of this Lemma is given in \Cref{lemmaProof}.
 
Now that we know what the parameterised Hamiltonian is in the fermionic space, which corresponds to our QCA, let us consider the Heisenberg evolution of operators in the fermionic picture. For quadratic Hamiltonians like above, which conserve particle number (not just fermionic parity), we have the following results, where the proof is given in \Cref{app:Ev_localQ}.

\begin{lemma}[Linear evolution under local quadratic Hamiltonians]\label{lemm:Ev_localQ} Let $H$ be a quadratic fermionic Hamiltonian on $n$ nodes of the form
\begin{equation}
    H = \sum_{j,k=1}^n h_{jk}a^{\dagger}_ja_k,
\end{equation}
where $h_{jk} \in \mathbb{C}$, $h=h^{\dagger}$, and $H$ is local (i.e., $h_{jk}=0$ if $|j-k|\gg 1$). Then the time-evolved annihilation operator in the Heisenberg picture satisfies
\begin{equation}
\label{eq:heisenberg}
a_j(\tau)=e^{iH\tau}a_j(0)e^{-iH\tau}=\sum_{k=1}^n A_{j,k}a_k(0),
\end{equation}
where $A=e^{-ih\tau}$.
\end{lemma}

Observe that in the solution above, \( A \in SU(n) \). Returning to the QCA evolution, recall that the transition functions act only on pairs of neighbouring cells, as described in \Cref{eq:TransitionFunc}, and in the fermionic representation are given by \Cref{eq:FermQCA}. The full evolution is generated by alternating applications of two local Hamiltonians:
\begin{equation}
\label{eq:Hm1}
H_e = \sum_{k=0}^{n/2-1} h^{(1)}_{2k}, \quad H_o = \sum_{k=0}^{n/2-1} h^{(2)}_{2k+1},
\end{equation}
where \( H_e \) corresponds to the Hamiltonian that generates the even-layer unitary \( W_0 \), and \( H_o \) to its odd-layer counterpart \( W_1 \), when \( \alpha = 0 \). Thus, the Heisenberg evolution of the annihilation operators takes the linear form given in \Cref{eq:heisenberg}, and as a result, we arrive at the following solution.

\begin{lemma}\label{lem:heisen_sol}
The solutions of the Heisenberg equations of motion \Cref{eq:heisenberg} for the fermionic annihilation operators, under the Hamiltonians \Cref{eq:Hm1}, are given by
\begin{equation}
A_e =
e^{i\phi_1} \, I_{n/2} \otimes 
\begin{pmatrix}
a_1 & b_1 \\
-b_1^{*} & a_1^{*}
\end{pmatrix},
\end{equation}
where $I_n$ represents a $n\times n$ identity matrix and
\begin{equation}
A_o = X_n A X_n^{\dagger},
\end{equation}
where
\begin{equation}
A =
e^{i\phi_2} \, I_{n/2} \otimes 
\begin{pmatrix}
a_2 &  b_2 \\
-b_2^{*} & a_2^{*}
\end{pmatrix},
\end{equation}
with
\begin{align}
a_j &= \cos \Omega_j + i \cos \beta_j \, \sin \Omega_j, \nonumber \\
b_j &= ie^{i \theta_j} \sin \Omega_j \, \sin \beta_j,
\end{align}
for \( j = 1, 2 \). Here,
\begin{equation}
X_m = \sum_{j=0}^{m-1} \ket{j} \bra{j \oplus_m 1},
\end{equation}
is the unitary cyclic shift matrix.
\end{lemma}

In the lemma above, where the proof is given in \Cref{app:Ev_localQ}, the parameters $\{\Omega,\beta,\theta\}\in[0,2\pi)$ are freely adjustable. Since $\phi$ is a global phase, we henceforth ignore it. We can now bring the following results from the lemma above, where the proof can also be contemplated in \Cref{app:Ev_localQ}.

\begin{corollary}
\label{cor:QCA_Ev}
Consider a local fermionic QCA with a single time-step evolution operator given by
\begin{equation}
    \mathbb{A} = A_o A_e,
\end{equation}
where \( A_o \) and \( A_e \) are defined in \Cref{lem:heisen_sol}. Under periodic boundary conditions and after \( t \) time steps, the evolution operator takes the form
\begin{equation}
    \mathbb{A}^t=[F_{n/2}\otimes I_2] \Bigg[\sum_{k=0}^{n/2-1}\ket{k}\bra{k}\otimes [M(k)]^t\Bigg] [F^{\dagger}_{n/2}\otimes I_2].
\end{equation}
where $F_m$ is discrete Fourier transform and
\begin{equation}
\label{eq:mk}
M(k)=\left(\begin{array}{cc}a_2^{*}a_1 +b_2^*b_1^*e^{-i 4\pi k/n} &a^{*}_2b_1  - b_2^{*}a_1^* e^{-i4\pi k/n} \\a_1b_2 e^{i 4\pi k/n} -a_2b^{*}_1& b_2b_1 e^{i 4\pi k/n} + a_2a_1^{*}\end{array}\right).
\end{equation}
\end{corollary}

Now, let the initial state of the system in the qubit representation consist of a string of $n-m$ zeros and $m$ ones, where the ones occur at locations in the set $S=\{s_j\}_{j=1}^m \subset \{1,\cdots,n\}$. Then, in the fermionic picture, the initial state is
\begin{equation}
\ket{\Psi(0)}=a^{\dagger}_{s_m}\ldots a^{\dagger}_{s_1}\ket{vac},    
\end{equation}
where $\ket{vac}$ is the vacuum state with no particles. Since we have the for $t$ time steps of the PUQCA in the Heisenberg picture given by 
\begin{equation}
 a_j(t)=\sum_{c}(\mathbb{A}^{t})_{jc}\,a_c(0).   \end{equation}
We can then compute $n_j(t)$,  the occupation number at the mode $j$, i.e.,
\begin{align}
\langle a_j^\dagger(t)a_j(t)\rangle
&=\langle\Psi(0)|\,a_j^\dagger(t)a_j(t)\,|\Psi(0)\rangle\nonumber\\
&=\sum_{\ell,c}(\mathbb{A}^{t})_{j\ell}^*(\mathbb{A}^{t})_{jc}\,
\langle{\rm vac}|a_{s_1}\cdots a_{s_m}\,a_\ell^\dagger a_c\,
a_{s_m}^\dagger\cdots a_{s_1}^\dagger|{\rm vac}\rangle.
\end{align}
We now have to look more carefully at the expression above. First, let us consider the case where $c\notin S$. In that scenario, we do not have a particle in mode $c$, so $a_c\ket{\Psi(0)}=0$. Now in the case where $c \in S$, we can use the anticommutation relation, $\{a_l^{\dagger},a_m\}=\delta_{l,m}I$, so that
\begin{align}
a_ka^{\dagger}_{s_m}\ldots a^{\dagger}_{s_{j+1}}a^{\dagger}_{s_j=c}a^{\dagger}_{s_{j+1}}\dots a^{\dagger}_{s_1}\ket{vac} &= (\pm)a^{\dagger}_{s_m}\ldots a^{\dagger}_{s_{j+1}}a_ca^{\dagger}_{s_j=c}a^{\dagger}_{s_{j+1}}\dots a^{\dagger}_{s_1}\ket{vac}\nonumber\\ 
&= (\pm)a^{\dagger}_{s_m}\ldots a^{\dagger}_{s_{j-1}}a^{\dagger}_{s_{j+1}}\dots a^{\dagger}_{s_1}\ket{vac}.    
\end{align}
The same idea can be done for the elements of the dual space, so we will only get nonzero results if $l=c$, which then gives
\begin{equation}
\langle a_j^\dagger(t)a_j(t)\rangle
=\sum_{c\in S}\big|(\mathbb{A}^{t})_{jc}\big|^2.
\end{equation}
Notice that the occupation number can be written using the projector onto the occupation space of the initial state $P_S=\sum_{c\in S}\ket{c}\bra{c}$ and follows
\begin{equation}
\langle a_j^\dagger(t)a_j(t)\rangle
 = \big(\mathbb{A}^{t}P_S \mathbb{A}^{t\dagger}\big)_{jj}.
\end{equation}

In the context of the PUQCA the projector onto the occupation number can be expressed in the pair $(c_s,c_p)$, i.e.
\begin{equation}
\label{eq:Proj_cc}
P_S =\sum_{(c_s,c_p) \in S}\ket{c_s}\bra{c_s}\otimes \ket{c_p}\bra{c_p},   
\end{equation}
where $c_s \in \mathbb{Z}_{n/2}$ and $c_p \in \mathbb{Z}_{2}$. 
In the context of the PUQCA, this corresponds to a system with $n/2$ cells, each containing two subcells. Thus, a label such as $(4,0)$ denotes an excitation in cell $4$ at the \emph{left} (or ``even'') subcell, while $(4,1)$ denotes an excitation in cell $4$ at the \emph{right} (or ``odd'') subcell. 

Since we are interested in bringing this mapping to the contest of measuring the excitation at the qubit $p$, in the fermionic picture it is equivalent to computing the occupation number, i.e. $Pr(1_{(j_s,j_p)})=\langle n_{(j_s,j_p)}(t)\rangle$, where $j_s\in \{0,\cdots,n/2-1\}$ tell which cell we are looking the meansurement and $j_p\in \{0,1\}$ tells us what is the subcell.  In particular, for $(0,0)$, which we will just say as the probability at the qubit 0 to keep with the same convention used in before the mapping to the fermionic case, is  
\begin{equation}
\label{eq:prob_eff}
Pr(1_{(0,0)}) = \frac{4}{n^2}\sum_{k,k'=0}^{n/2-1}  \sum_{(c_s,c_p)\in S} e^{-4\pi i c_s(k-k') /n}\bra{0}[M(k)]^t\ket{c_p}\bra{c_p} [M(k')]^{\dagger t}\ket{0},   
\end{equation}
where the full derivation of the expression above is given in \Cref{app:Prob}.

\subsection{Results II}

We now report the results found using the fermionic PUQCA. Despite the efficient method for treating larger system sizes, the computation required to find the QCA rule is still quite demanding due to the exponential number of initial states as a function of the system size that need to be tested to obtain the PUQCA, which solves the density classification task.

In the results reported here, we considered the same parameters as in the previous evolutionary algorithm, i.e., $ p_m=0.36$ and $ \sigma=0.45$, with $\delta=0$, and the measurement point $p=0$. The population size also varied between 100 and 180, with a maximum of 100 generations and, the angle parametrisation used in Equation \ref{eq:mk} is
\begin{equation}
 a_j = \exp(i\xi_j)\cos(\theta_j), \quad b_j= \exp(i*\gamma_j)\sin(\theta_j),
\end{equation}
for $j=1,2$. 

In the tables below, we show the results for the PUQCA in the simulatable regime from 4 to 8 cells in Table \ref{tb:res1_pt2}, where the rules resolve perfectly when the number of steps is equal to half the number of cells and for Table \ref{tb:res2_pt2}, which has some considerations. For the case with 10 cells, rather than $t=5$, we only found a solution for $ t=6$. In the case of 12 cells, rather than $t=6$, we only found solutions for $t=13$. However, with 14 cells, we could not find a PUQCA rule that solves perfectly for all initial conditions; we only found a rule with a fitness value of 0.9722, which occurs at $t=15$. It is also important to stress that the parameters in Table \ref{tb:res2_pt2} were tested in Equation \ref{eq:prob_eff} and found to agree.

\begin{table}[h!]
\centering
 \begin{tabular}{||c c c ||} 
 \hline
 4 Cells & 6 Cells & 8 Cells \\ [0.5ex] 
 \hline\hline
 $\theta_1= 0.3692$, $\xi_1= 0.3158$ & $\theta_1 = 0.8452$, $\xi_1 = 0.1676$ & $\theta_1=0.5985$, $\xi_1 = 1.7330$  \\ 
 $\alpha_1=0.1242$~ & $\alpha_1 =  0.0163$~ &  $\alpha_1 = 0.2285$~  \\
   &  &   \\
 $\theta_2= 0.2891$, $\xi_2=0.2193$ & $\theta_2 = 0.5060$, $\xi_2 = 0.4904$ & $\theta_2 = 0.8025$, $\xi_2 = 0.3008$  \\
 $\alpha_2=0.5550$~ &   $\alpha_2 = 6.1718$~ &  $\alpha_2 = 0.6899$ \\ [1ex] 
 \hline
 \end{tabular}
 \caption{Angle configurations that solve the classification problem within the classically simulable PUQCA from $4$ to $8$ cells. \label{tb:res1_pt2}}
\end{table}

\begin{table}[h!]
\centering
 \begin{tabular}{||c c c ||} 
 \hline
 10 Cells & 12 Cells & 14 Cells \\ [0.5ex] 
 \hline\hline
 $\theta_1= 0.3963$, $\xi_1= 0.8772$ & $\theta_1 = 0.7089$, $\xi_1 =  0.9207$ & $\theta_1=0.4592$, $\xi_1 = 0.7973$  \\ 
 $\alpha_1=0.7499$~ & $\alpha_1 = 0.6102$~ &  $\alpha_1 = 0.7148$~  \\
   &  &   \\
 $\theta_2= 0.4026$, $\xi_2= 0.6149$ & $\theta_2 = 0.7453$, $\xi_2 = 0.4214$ & $\theta_2 = 0.4311$, $\xi_2 = 0.8624$  \\
 $\alpha_2= 1.1856$~ &   $\alpha_2 =0.4745$~ &  $\alpha_2 = 0.8936$ \\ [1ex] 
 \hline
 \end{tabular}
 \caption{Angle configurations that solve the classification problem within the classically simulable PUQCA from $10$ to $12$ cells, but only with a large accuracy for the case with 14 cells. \label{tb:res2_pt2}}
\end{table}

Now, in the multiple-solution case—consistent with Table \ref{tab:fitmult_simReg_A}—fitting across multiple system sizes within the six-parameter PUQCA yields low accuracy, comparable to that obtained with the full PUQCA.

\begin{table}[H]
\centering
\begin{tabular}{c|c}
\toprule
\textbf{System size $n$} & \textbf{Fitness $F$} \\
\midrule
4  & 1.0000 \\
6  & 0.9091 \\
8  & 0.8925 \\
10 & 0.9896 \\
12 & 0.7995 \\
14 & 0.7763 \\
\bottomrule
\end{tabular}
\caption{
Single PUQCA rule, in the classically simulable regime, applied across sizes; perfect performance only for $n=4$ and a reasonable accuracy for $n=10$. Parameters:
$\theta_1 = 0.4949$, $\xi_1 = 0.7195$, $\gamma_1 =  0.0337$;
$\theta_2 = 0.4790$, $\xi_2 = 1.0300$, $\gamma_2 = 0.4848$.
\label{tab:fitmult2_simReg_A}}
\end{table}

\section{Conclusion}\label{sec:conc}

Despite the non-existing stationary states for the PUQC, we brought a new success condition that can be used for the PUQCA to solve the DCT. We could find PUQCA rules using the evolutionary algorithm for each system size, along with the number of PUQCA updates, $\mathcal{E}(W_1,W_0)$, with exactly half the number of cells considered. We also found rules that could solve with a higher accuracy for multiple system sizes, with the number of PUQCA updates being half the number of cells. Moreover, we show that the PUQCA admits a simple mapping to a classically simulable regime and that, under an alternative success criterion, it provides a viable route to solving the density classification task. Within the classically simulable PUQCA, we can find rules that perfectly solve the DCT from four to twelve cells, without necessarily following the number of updates as in the full regime PUQCA. We also could find a rule that solves well when $n=14$ and a rule that solves for multiple system sizes seems also to be possible.

Whether a single rule can solve the task uniformly across multiple system sizes remains an open question; we have not found an analytic construction. On the other hand, there exist generalized Dicke states that yield the correct outcome for our DCT success criterion, so such solutions exist in principle. Although we have not tested this here, it is natural to ask whether PUQCA rules that succeed across sizes effectively prepare states close to Dicke states. Thus, alternatively, we could consider a fitness function that finds angles, for a fixed number of time steps, which maps the initial classical configuration to a state that closely approximates the generalized Dick states with the same number of excitations.

The computational challenge arises not only in the non–classically simulable regime of the PUQCA, but also from the exponential growth in the number of initial configurations that must be tested to certify a rule for the DCT. Following the spirit of~\cite{wolz2008very}, we could instead evaluate fitness on a suitably chosen \emph{sample} of initial configurations, trading exact certification for scalability. This sampling strategy could substantially extend the search space and reveal additional PUQCA rules that solve (or closely approximate) the DCT beyond what we explored here. 

Since one array PUQCA, at one time step, compiles to a shallow 1D nearest–neighbour circuit (two-site unitaries per layer, Figure \ref{fig:PUQCA_rep}), this also suggests a convenient benchmarking task for current quantum hardware experiments and for the DCT setting.

The density classification task with PUQCA—a global majority decision from strictly local updates—has natural touchpoints with quantum error correction (QEC) \cite{QCADenClass}. In repetition and other majority-vote decoders, the logical bit is inferred by a global majority over physical readouts, implemented via local parity/syndrome measurements followed by a (possibly local) aggregation rule; this is the DCT archetype in 1D form \cite{NielsenChuang}. More broadly, cellular-automaton style local decoders have been proposed for topological codes, where syndromes evolve under local update rules that aim to “concentrate” defects and drive a global correction without nonlocal matching—see Harrington’s local decoder for the surface code \cite{HarringtonThesis}, CA decoders for topological memories \cite{Herold2015PRL}, and early foundations in topological QEC \cite{DennisKitaevLandahlPreskill}. In these settings, a quantum formulation for DCT-like primitive (majority/success criterion under locality and homogeneity) is directly relevant to \textit{syndrome extraction} pipelines and \textit{fault-tolerant, low-latency decoding}, where communication and classical controller overhead are constrained.

\section*{Acknowledgements}
P.C.S.C. is grateful to Fernando de Melo, his former PhD supervisor, for proposing the initial direction of this work and for insightful early guidance.
P.P.B. thanks the Brazilian agencies CNPq for the research grant PQ 303356/2022-7, and CAPES for the research grants STIC-AmSud  88881.694458/2022-01 and Mackenzie-PrInt 88887.310281/2018-00. G.K.B. acknowledges support from the Australian Research Council (ARC) Centre of Excellence for Engineered Quantum Systems (EQUS) (Grant No.
CE170100009).


\appendix

\section{The Majorana representation of PUQCA}
\label{lemmaProof}
In this section, we first prove \Cref{lemm:Ising_QCA}, i.e., that the time evolution under the Ising Hamiltonian given in \Cref{eq:Isin_QCA} corresponds to the time evolution for the SU(2) PUQCA. Next, we prove \Cref{lem:ferm_halm} that shows how the Ising PUQCA can be mapped to the fermionic space.

\subsection{Proof of \Cref{lemm:Ising_QCA}}
\begin{proof}
Let us start by looking at $h_k$ for a given fixed point $k$ in the lattice and then at its matrix format,
\begin{equation}
\label{eq:Isin_QCA}
    h=\frac{J}{4}e^{i \theta}\sigma^{+} \otimes \sigma^{-} + \frac{J}{4}e^{-i \theta}\sigma^{-} \otimes \sigma^{+}+h_1\sigma^{z}\otimes I_2+h_2I_2\otimes \sigma_{z},
\end{equation}
Here, we explicitly wrote the tensor product between the operators as we want to look at its matrix format. Let us start with the matrix format for every two components and then add them. For the first two, we have
\begin{equation}
 \frac{J}{4}e^{i \theta}\sigma^{+} \otimes \sigma^{-} + \frac{J}{4}e^{-i \theta}\sigma^{-} \otimes \sigma^{+}=
\begin{bmatrix}0 & 0 & 0 & 0\\
0 & 0 &  Je^{i\theta}& 0\\
0 & Je^{-i\theta}& 0 & 0\\
0 & 0 & 0 & 0
\end{bmatrix},   
\end{equation}
and for the two others
\begin{equation}
h_1\sigma^{z}\otimes I_2+h_2I_2\otimes \sigma^{z}=
\begin{bmatrix}h_1+h_2 & 0 & 0 & 0\\
0 & h_1-h_2 &  0& 0\\
0 & 0& -h_1+h_2 & 0\\
0 & 0 & 0 & -h_1-h_2
\end{bmatrix},
\end{equation}
therefore,
\begin{equation}
\label{eq:h}
    h=\begin{bmatrix}h_1+h_2 & 0 & 0 & 0\\
0 & h_1-h_2 &  Je^{i\theta}& 0\\
0 & Je^{-i\theta}& -h_1+h_2 & 0\\
0 & 0 & 0 & -h_1-h_2
\end{bmatrix}.
\end{equation}
Our next step is the computation of the unitary $W$ generated by the Hamiltonian $h$ acting in the two neighbors' qubits, i.e.,
\begin{equation}
\label{eq:UnitW}
    W=e^{-i\tau h}=I_4-i\tau h-\frac{\tau^2}{2}h^2+\frac{i\tau^{3}}{3!}h^3+\frac{\tau^{4}}{4!}h^4+\cdots,
\end{equation}
where $I_4$ is a $4 \times 4$ identity matrix. The matrix form of $W$ can be determined by computing the elements $h^2,h^3$, and $h^4$, as we will see next. These four operators are given by, 

\begin{equation}
h^2=\begin{bmatrix}(h_1+h_2)^2 & 0 & 0 & 0\\
0 & (h_1 - h_2)^2 + J^2 & 0& 0\\
0 & 0& (h_1 - h_2)^2 + J^2 & 0\\
0 & 0 & 0 & (h_1 + h_2)^2
\end{bmatrix},
\end{equation}

\begin{equation}
h^3=\begin{bmatrix}(h_1+h_2)^3 & 0 & 0 & 0\\
0 & (h_1 - h_2) ((h_1 - h_2)^2 + J^2) & e^{i\theta} J ((h_1 - h_2)^2 + J^2)& 0\\
0 & e^{-i \theta} J ((h_1 - h_2)^2 + J^2)& (-h_1 + h_2) ((h_1 - h_2)^2 + J^2) & 0\\
0 & 0 & 0 & -(h_1 + h_2)^2
\end{bmatrix},
\end{equation}

\begin{equation}
h^4=\begin{bmatrix}(h_1+h_2)^4 & 0 & 0 & 0\\
0 & ((h_1 - h_2)^2 + J^2)^2 & 0& 0\\
0 & 0& ((h_1 - h_2)^2 + J^2)^2 & 0\\
0 & 0 & 0 & -(h_1 + h_2)^4
\end{bmatrix}.
\end{equation}
Now, if we make use of these matrices and use them in \Cref{eq:UnitW} it is easy to show that the first and the last elements of $W$ are given by 
\begin{equation}
    W_{1,1}=e^{i\phi}, \quad  W_{4,4}=e^{-i\phi},
\end{equation}
where $\phi=(h_1+h_2)\tau$. From the even powers elements of $h$ we can write
$\cos\Omega$, where 
\begin{equation}
    \Omega =\tau \sqrt{(h_1 - h_2)^2 + J^2}. 
\end{equation}
Now let us concentrate in the odd powers of $h$, but only for the element $W_{2,2}$,
\begin{align}
W_{2,2}&=-i (h_1-h_2)\tau+i(h_1 - h_2) \left[(h_1 - h_2)^2 + J^2\right]\frac{\tau^3}{3!}-i(h_1 - h_2) \left[(h_1 - h_2)^2 + J^2\right]^2\frac{\tau^5}{5!}+\cdots \nonumber\\
&=-i(h_1-h_2)\Bigg\{ \tau-\left[(h_1 - h_2)^2 + J^2\right]\frac{\tau^3}{3!} + \left[(h_1 - h_2)^2 + J^2\right]^2\frac{\tau^5}{5!}+\cdots\Bigg\}\nonumber\\
&=\frac{-i(h_1-h_2)}{\sqrt{(h_1 - h_2)^2 + J^2}}\Bigg\{ \sqrt{(h_1 - h_2)^2 + J^2} \tau-\left[(h_1 - h_2)^2 + J^2\right]^{3/2}\frac{\tau^3}{3!}\nonumber\\ 
&\quad+ \left[(h_1 - h_2)^2 + J^2\right]^{5/2}\frac{\tau^5}{5!}+\cdots\Bigg\}\nonumber\\
&=-i(h_1-h_2)\frac{\sin{\Omega}}{\Omega}.
\end{align}
Doing the same for the other two elements, we can conclude that    
\begin{equation}
W=\begin{bmatrix}e^{i\phi} & 0 & 0 & 0\\
0 & \cos{\Omega}-i(h_1-h_2)  \text{ sinc }{\Omega} &  -ie^{i\theta}J \text{ sinc }{\Omega}& 0\\
0 & -ie^{-i\theta}J\text{ sinc }{\Omega}& \cos{\Omega}+i(h_1-h_2) \text{ sinc }{\Omega} & 0\\
0 & 0 & 0 & e^{-i\phi}
\end{bmatrix},
\end{equation}
where $\text{ sinc }{\Omega}=1/\Omega \sin{\Omega}$ or 
\begin{equation}
W=\begin{bmatrix}e^{i\phi} & 0 & 0 & 0\\
0 & \cos{\Omega}-i\cos{\beta} \sin{\Omega} &  -ie^{i\theta}\sin\beta\sin{\Omega}& 0\\
0 & -ie^{-i\theta}\sin\beta\sin{\Omega}& \cos{\Omega}+i\cos{\beta} \sin{\Omega} & 0\\
0 & 0 & 0 & e^{-i\phi}
\end{bmatrix},
\end{equation}
where $\tan \beta = J/(h_1-h_2)$.
\end{proof}

\subsection{Proof of \Cref{lem:ferm_halm}}
\begin{proof}
The first thing that should be noticed is the following anti-commutation relations
\begin{equation}
\label{eq:new_anti}
    \{a_j^{\dagger},e^{i\pi n_j}\} = 0,\quad  \{a_j,e^{i\pi n_j}\} = 0,
\end{equation}
where $n_j=a_j^{\dagger}a_j$. The anti-commutation relations above can be derived from the fermionic algebra
\begin{equation}
\label{eq:anti_com}
    \{a_l^{\dagger},a_m^{\dagger}\} = 0,\quad  \{a_l,a_m\} = 0,\quad \{a_l^{\dagger},a_m\} = \delta_{l,m}I.
\end{equation}
Now, from \Cref{eq:stringOp} and \Cref{eq:new_anti} we can derive the main anti-commutation relations to complete the proof,
\begin{equation}
\{e^{i\phi_j}, a_l^{(\dagger)}\} = 0\quad (l<j),\quad \left[e^{i\phi_j}, a_l^{(\dagger)}\right] = 0\quad (l\geq j).
\end{equation}
The mapping from the $n$ qubits on a 1D lattice with spin operators to the $n$ spinless fermionic modes can be quickly done, as we will show next. Starting with the first component of \Cref{eq:QCAhamilt}, we have from \Cref{eq:mapping} that
\begin{align}
    \frac{J}{4}e^{i\theta} \sigma^+_k \sigma^-_{k+1} &= Je^{i\theta}a_k^\dagger e^{i \phi_k}a_{k+1} e^{-i \phi_{k+1}}\nonumber\\
    &=Je^{i\theta}a_k^\dagger e^{i (\phi_k-\phi_{k+1})}a_{k+1}\nonumber\\
    &=Je^{i\theta}a_k^\dagger e^{i\pi (\sum_{l=0}^{k-1}a_l^{\dagger}a_l-\sum_{l=0}^{k}a_l^{\dagger}a_l)}a_{k+1}\nonumber\\
    &=Je^{i\theta}a_k^\dagger e^{-i\pi n_k}a_{k+1}\nonumber\\
    &=Je^{i\theta}a_k^\dagger a_{k+1},
\end{align}
where in the last equality we used the following identity
\begin{equation}
    e^{-i\pi n_k} = I +2n_k. 
\end{equation}
The analogous calculation can be done for the component $\sigma^-_k \sigma^+_{k+1}$. Finally, the last two components of the Hamiltonian $h_k$, \Cref{eq:QCAhamilt}, can be checked easily by using the following identity $(\sigma_z)_j=1-2a^{\dagger}_ja_j$, which then completes our proof.
\end{proof}

\section{The classical simulable  PUQCA}\label{app:Ev_localQ}

In this section, we bring the proofs of \Cref{lem:heisen_sol,lemm:Ev_localQ} and \Cref{cor:QCA_Ev}. To that end, we need to derive the equation for the fermionic operator in the \textit{Heisenberg picture},
\begin{equation}
\frac{d}{dt}O(t) = i [H,O(t)],    
\end{equation}
where $O(t)$ is an operator Heisenberg picture and $H$ is the Hamiltonian.

\subsection{Proof of \Cref{lemm:Ev_localQ}}
\begin{proof}
 In the first case, we have the following Hamiltonian
\begin{equation}
 H = \sum_{l,m=1}^n h_{lm}a^{\dagger}_la_m.    
\end{equation}

So, our goal is to solve
\begin{equation}
\label{eq:heis}
\frac{d}{dt}a_{j}(t) = i\sum_{l,m=1}^n [h_{lm}a^{\dagger}_la_m ,a_j(t)].    
\end{equation}
We will compute the commutator in the Hamiltonian, considering the canonical anti-commutation relations given in \Cref{eq:anti_com}. At the end, what we need to compute is
\begin{align}
     [a^{\dagger}_la_m ,a_j] &= a^{\dagger}_la_ma_j - a_ja^{\dagger}_la_m\nonumber\\
     &= -a^{\dagger}_la_ja_m - a_ja^{\dagger}_la_m\nonumber\\
      &= -\delta_{lj}a_m + a_ja^{\dagger}_la_m - a_ja^{\dagger}_la_m.
\end{align}
Thus, by plugging the results above into \Cref{eq:heis} we get
\begin{equation}
    \frac{d}{dt}a_{j}(t) = -i\sum_{m=1}^n h_{jm}a_m(t),
\end{equation}
which is a linear system of ODE's, i.e.,
\begin{equation}
  \frac{d}{dt}\mathbf{a}(t) = ih\mathbf{a}(t),   
\end{equation}
where $\mathbf{a}(t)=[a_1(t),\cdots,a_n(t)]^T$. Therefore, the solution is given by 
\begin{equation}
 \mathbf{a}(t)  = e^{-iht}\mathbf{a}(0), 
\end{equation}
or in terms of its components, we have
\begin{equation}
 a_j(t)  = \sum_{k=1}^nA_{jk}a_k(0), 
\end{equation}
where
\begin{equation}
A =  e^{-iht}.  
\end{equation}

\subsection{Proof of \Cref{lem:heisen_sol}}

In the second case we have to solve the Heisenberg equation for 
\begin{equation}
 H_e = \sum_kh_{2k},  
\end{equation}
with 
\begin{equation}
h_k=Je^{i\theta} a_k a^{\dagger}_{k+1}+Je^{-i\theta} a_{k+1} a^{\dagger}_{k}+h_1(1-2a^{\dagger}_ka_k)+h_2(1-2a^{\dagger}_{k+1}a_{k+1}).
\end{equation}
In this case, we see that we have nontrivial solutions only in subspaces with two modes, i.e., we only have to consider the solutions for
\begin{align}
\frac{d}{dt}a_{j}(t) &= i[h_{j},a_j(t)]\nonumber\\  
\frac{d}{dt}a_{j+1}(t) &= i[h_{j},a_{j+1}(t)].
\end{align}
We will compute the commutator term by term in the Hamiltonian, considering the canonical anti-commutation relations given in \Cref{eq:anti_com}. For the first term in the Hamiltonian, we have
\begin{align}
Je^{i\theta} [a_ja^{\dagger}_{j+1},a_j] 
&= Je^{i\theta} \left(a_{j}a^{\dagger}_{j+1}a_j - a_j^2a^{\dagger}_{j+1}\right)\nonumber\\
&= -2Je^{i\theta}a^{\dagger}_{j+1}a^2_j\nonumber\\
&=0,
\end{align}
where we used the fact that $a_j^2=0$. Next, we have the term that involves $a_{j+1}a^{\dagger}_j$,
\begin{align}
Je^{-i\theta}  [a_{j+1}a^{\dagger}_{j},a_j] 
&= Je^{-i\theta} \left(a_{j+1}a^{\dagger}_{j}a_j - a_ja_{j+1}a^{\dagger}_{j} \right)\nonumber\\
&= Je^{-i\theta} \left(a_{j+1}a^{\dagger}_{j}a_j + a_{j+1}a_ja^{\dagger}_{j}\right)\nonumber\\
&=Je^{-i\theta}a_{j+1}.
\end{align}
We now move to the term that involves $1-2a_j^{\dagger}a_j$.
\begin{align}
h_1 [1-2a_{j}^{\dagger}a_{j},a_j] &=  -h_1[2a^{\dagger}_{j}a_{j},a_j] \nonumber\\
&=-2h_1\left(a^{\dagger}_{j}a^2_{j} - a_ja_j^{\dagger}a_j\right)\nonumber\\
&= 2h_1a_j.
\end{align}
Finally, we move to the analogous term, i.e., $1-2a_{j+1}^{\dagger}a_{j+1}$, which gives
\begin{equation}
h_2 [1-2a_{j+1}^{\dagger}a_{j+1},a_j] =0.  
\end{equation}
Thus, for the first mode we have
\begin{equation}
\frac{da_j(t)}{dt} =  iJe^{-i\theta}a_{j+1}(t) + 2ih_1a_j(t).    
\end{equation}
The same procedure has to be done for  $d a_{j+1}/dt$ that gives
\begin{equation}
\frac{da_{j+1}(t)}{dt} =  iJe^{i\theta}a_{j}(t) + 2ih_2a_{j+1}(t).    
\end{equation}
We then have a two-mode coupled ODE system to solve. If we define
\begin{equation}
\mathbf{a}^{(2)}_j(t) \coloneqq  \begin{bmatrix}a_{j}(t) \\
a_{j+1}(t)
\end{bmatrix}, 
\end{equation}
we have
\begin{equation}
 \frac{d}{dt} \mathbf{a}_j^{(2)}(t) =i M \mathbf{a}_j^{(2)}(t),   
\end{equation}
where
\begin{equation}
M =  \begin{bmatrix} 2h_1& Je^{-i\theta}\\
Je^{i\theta} & 2h_2
\end{bmatrix}.   
\end{equation}
So the solution is given by
\begin{equation}
 \mathbf{a}^{(2)}_j(t) = e^{iMt}\mathbf{a}^{(2)}_j(0).
\end{equation}
Now let us analyse the matrix $tM$. From the previous definitions, we have $\phi=(h_1+h_2)t$ and $\Omega = t\sqrt{(h_1-h_2)^2 + J^2}$. Considering both parameters, we can rewrite the matrix $M$ as
\begin{equation}
\label{eq:M}
tM = \phi I + tN,
\end{equation}
where 
\begin{equation}
N =  \begin{bmatrix} h_1 -h_2& Je^{-i\theta}\\
Je^{i\theta} & h_2- h_1
\end{bmatrix}.   
\end{equation}
To analyse $e^{iMt}$, we need to evaluate the powers of $M$, however, we can take advantage from the fact that 
\begin{equation}
    [\phi I, t N] =0,
\end{equation}
So
\begin{equation}
e^{iMt} = e^{i(\phi I + tN)}   = e^{i\phi I}e^{tN}.  
\end{equation}
Thus, from the simple algebraic calculation above, we notice that for the exponentiation of $N$ we get to the same submatrix given in \Cref{eq:h}, the elements of the subspace of one excitation. Therefore, we have,
\begin{equation}
e^{tN}=\begin{bmatrix}
 \cos{\Omega}+i\cos{\beta} \sin{\Omega} &  ie^{i\theta}\sin\beta\sin{\Omega}\\
 ie^{-i\theta}\sin\beta\sin{\Omega}& \cos{\Omega}-i\cos{\beta} \sin{\Omega} 
\end{bmatrix},
\end{equation}
where $\tan \beta = J/(h_1-h_2)$.   
\end{proof}
Therefore, the evolution of the two-mode system is given by
\begin{equation}
 \mathbf{a}^{(2)}_j(t) = 
 e^{i\phi}
 \begin{bmatrix}
 \cos{\Omega} + i \cos{\beta} \sin{\Omega} & i e^{i\theta} \sin{\beta} \sin{\Omega} \\
 i e^{-i\theta} \sin{\beta} \sin{\Omega} & \cos{\Omega} - i \cos{\beta} \sin{\Omega} 
 \end{bmatrix}
 \mathbf{a}^{(2)}_j(0).
\end{equation}
The expression above applies to an isolated two-mode system, which remains invariant under the evolution. Extending this to a chain of $n$ nodes, the global evolution reads
\begin{equation}
 \mathbf{a}(t) =
 e^{i\phi}
 I_{n/2} \otimes 
 \begin{bmatrix}
 \cos{\Omega} + i \cos{\beta} \sin{\Omega} & i e^{i\theta} \sin{\beta} \sin{\Omega} \\
 i e^{-i\theta} \sin{\beta} \sin{\Omega} & \cos{\Omega} - i \cos{\beta} \sin{\Omega} 
 \end{bmatrix}
 \mathbf{a}(0).
\end{equation}

\subsection{Proof of \Cref{cor:QCA_Ev}}

\begin{proof}
From the definition of the unitary cyclic shift matrix, we that the following relation is valid when $n$ is even
\begin{equation}
X_n=X_{n/2}\otimes \left(\begin{array}{cc}0 & 0 \\1 & 0\end{array}\right)+I_{n/2}\otimes \left(\begin{array}{cc}0 & 1 \\0 & 0\end{array}\right),    
\end{equation}
In the matrix above, we have for $X_2$ just the Pauli-$X$. We can now use the decomposition above to express $A_o = X_n A X_n^{\dagger}$, where 
\begin{equation}
A =
 I_{n/2} \otimes 
\begin{pmatrix}
a_2 &  b_2 \\
-b_2^{*} & a_2^{*}
\end{pmatrix},
\end{equation}
where we have ignored the global phase. Thus, we can write
\begin{align}
X_n A X_n^{\dagger} &=
\Bigg[X_{n/2}\otimes \left(\begin{array}{cc}0 & 0 \\1 & 0\end{array}\right) + I_{n/2}\otimes \left(\begin{array}{cc}0 & 1 \\0 & 0\end{array}\right)\Bigg]  \, I_{n/2} \otimes 
\begin{pmatrix}
a_2 & b_2 \\
-b_2^{*} & a_2^{*}
\end{pmatrix}\nonumber\\
&\quad\times\Bigg[X_{n/2}\otimes \left(\begin{array}{cc}0 & 0 \\1 & 0\end{array}\right)+I_{n/2}\otimes \left(\begin{array}{cc}0 & 1 \\0 & 0\end{array}\right)\Bigg]^{\dagger}\nonumber\\
&=\Bigg[X_{n/2}\otimes\begin{pmatrix}
0 & 0\\
a_2 & b_2 
\end{pmatrix} +I_{n/2}\otimes \left(\begin{array}{cc}-b_2^{*} & a_2^{*} \\0 & 0\end{array}\right)\Bigg]\Bigg[X^{\dagger}_{n/2}\otimes \left(\begin{array}{cc}0 & 1 \\0 & 0\end{array}\right)+I_{n/2}\otimes \left(\begin{array}{cc}0 & 0 \\1 & 0\end{array}\right)\Bigg]\nonumber\\
&= I_{n/2}\otimes\begin{pmatrix}
a_2^{*} & 0\\
0 & a_2 
\end{pmatrix} + X_{n/2}\otimes\begin{pmatrix}
0 & 0\\
b_2 & 0 
\end{pmatrix} +  X^{\dagger}_{n/2}\otimes\begin{pmatrix}
0 & -  b_2^{*}\\
0 & 0 
\end{pmatrix}.
\end{align}
We can now use the results above to compute the full state, i.e.,
\begin{align}
 \mathcal{A} &= A_o A_e\nonumber\\
 &= \Bigg[I_{n/2}\otimes\begin{pmatrix}
a_2^{*} & 0 \\
0 & a_2 
\end{pmatrix} + X_{n/2}\otimes\begin{pmatrix}
0 & 0\\
b_2 & 0 
\end{pmatrix} +  X^{\dagger}_{n/2}\otimes\begin{pmatrix}
0 & -  b_2^{*}\\
0 & 0 
\end{pmatrix}\Bigg] I_{n/2} \otimes 
\begin{pmatrix}
a_1 & b_1 \\
-b_1^{*} & a_1^{*}
\end{pmatrix}\nonumber\\
&=I_{n/2}\otimes\begin{pmatrix}
a_2^{*}a_1 & a^{*}_2b_1 \\
-a_2b_1^{*} & a_2a_1^{*} 
\end{pmatrix} + X_{n/2}\otimes\begin{pmatrix}
0 & 0\\
b_2a_1 & b_2b_1 
\end{pmatrix} +  X^{\dagger}_{n/2}\otimes\begin{pmatrix}
b_2^{*}b_1^{*} & -  b_2^{*}a_1^{*}\\
0 & 0 
\end{pmatrix},
\end{align}
From the properties of the cyclic operator $U$ its eigenvalues are $\lambda_k =\exp{(2\pi ik/n)}$, with the eigenvectors given by the Fourier basis state, i.e.,
\begin{equation}
 \ket{k} = \frac{1}{\sqrt{n}}\sum_{j,k=0}^{n-1} e^{2\pi i j k/n}\ket{j}, 
\end{equation}
so
\begin{equation}
X_{n/2} = F_{n/2} \text{diag}\left(e^{4\pi i k/n}\right) F^{\dagger}_{n/2},  
\end{equation}
where
\begin{equation}
F_n=\frac{1}{\sqrt{n}}\sum_{j,k=0}^{n-1} e^{2\pi i j k/n}\ket{j}\bra{k}   
\end{equation}
is the discrete Fourier matrix.
Thus, we can bring $A$ to a simpler diagonal form 
\begin{align}
\mathcal{A} &=  (F_{n/2}\otimes I_2)\Bigg[ I_{n/2}\otimes\begin{pmatrix}
a_2^{*}a_1 & a^{*}_2b_1 \\
-a_2b_1^{*} & a_2a_1^{*}
\end{pmatrix} + \sum_{k=0}^{n/2-1}e^{4\pi i k/n}\ket{k}\bra{k} \otimes\begin{pmatrix}
0 & 0\\
b_2a_1 & b_2b_1 
\end{pmatrix}\nonumber\\ 
&\quad +  \sum_{k=0}^{n/2-1}e^{-4\pi i k/n}\ket{k}\bra{k}\otimes\begin{pmatrix}
b_2^{*}b_1^{*} & -  b_2^{*}a_1^{*}\\
0 & 0 
\end{pmatrix} \Bigg] (F^{\dagger}_{n/2}\otimes I_2) \nonumber\\
&=[F_{n/2}\otimes I_2] \Bigg[\sum_{k=0}^{n/2-1}\ket{k}\bra{k}\otimes [M(k)]\Bigg] [F^{\dagger}_{n/2}\otimes I_2],
\end{align}
where
\begin{equation}
M(k)=\left(\begin{array}{cc}a_2^{*}a_1 +b_2^*b_1^*e^{-i 4\pi k/n} &a^{*}_2b_1  - b_2^{*}a_1^* e^{-4i\pi k/n} \\a_1b_2 e^{i 4\pi k/n} -a_2b^{*}_1& b_2b_1 e^{i 4\pi k/n} + a_2a_1^{*}\end{array}\right).
\end{equation}
\end{proof}

\section{PUQCA probabilities in the Majorana representation }
\label{app:Prob}

To proceed with the computation for the probabilities in the Fermionic picture, let us now use the results from \Cref{cor:QCA_Ev} together 
\begin{align}
\mathcal{A}^{t}P_S \mathcal{A}^{t\dagger} & =[F_{n/2}\otimes I_2] \Bigg[\sum_{k=0}^{n/2-1}\ket{k}\bra{k}\otimes [M(k)]^t\Bigg] [F^{\dagger}_{n/2}\otimes I_2] \sum_{(c_s,c_p)\in S}\ket{c_s}\bra{c_s}\otimes \ket{c_p}\bra{c_p} [F_{n/2}\otimes I_2]\nonumber\\
&\quad\times \Bigg[\sum_{k'=0}^{n/2-1}\ket{k'}\bra{k'}\otimes [M(k')]^{\dagger t}\Bigg] [F_{n/2}^{\dagger}\otimes I_2].
\end{align}

To proceed with this calculation, let us break it into parts. First, using the expression for the Fourier transformation, 
\begin{equation}
F_n=\frac{1}{\sqrt{n}}\sum_{j,l=0}^{n-1} e^{2\pi i j l/n}\ket{j}\bra{l},   
\end{equation}
we have
\begin{align}
[F^{\dagger}_{n/2}\otimes I_2] \sum_{(c_s,c_p)\in S}\ket{c_s}\bra{c_s}\otimes \ket{c_p}\bra{c_p} [F_{n/2}\otimes I_2] &= \left(\frac{1}{n/2}\sum_{l=0}^{n/2-1} \sum_{(c_s,c_p)\in S} e^{-4\pi i l c_s/n}\ket{l}\bra{c_s} \otimes \ket{c_p}\bra{c_p}\right)\nonumber\\
&\quad \times\sum_{j_1,l_1=0}^{n/2-1} e^{4\pi i j_1 l_1/n}\ket{j_1}\bra{l_1}\otimes I_2\nonumber\\
&= \frac{1}{n/2}\sum_{l,j_1=0}^{n/2-1} \sum_{(c_s,c_p)\in S} e^{-4\pi i (l-l_1) c_s/n}\ket{l}\bra{l_1} \otimes \ket{c_p}\bra{c_p}.
\end{align}
Thus,
\begin{align}
\mathcal{A}^{t}P_S \mathcal{A}^{t\dagger} &= [F_{n/2}\otimes I_2] \Bigg[\sum_{k=0}^{n/2-1}\ket{k}\bra{k}\otimes [M(k)]^t\Bigg] \left(\frac{1}{n/2}\sum_{j_1,l=0}^{n/2-1} \sum_{(c_s,c_p)\in S} e^{-4\pi i (l-l_1) c_s/n}\ket{l}\bra{l_1} \otimes \ket{c_p}\bra{c_p}\right)\nonumber\\
&\quad\times \Bigg[\sum_{k'=0}^{n/2-1}\ket{k'}\bra{k'}\otimes [M(k')]^{\dagger t}\Bigg] [F_{n/2}^{\dagger}\otimes I_2]\nonumber\\
&=\frac{2}{n} [F_{n/2}\otimes I_2]\sum_{k,k'=0}^{n/2-1}  \sum_{(c_s,c_p)\in S} e^{-4\pi i (k-k') c_s/n}\ket{k}\bra{k'} \otimes [M(k)]^t\ket{c_p}\bra{c_p} [M(k')]^{\dagger t}[F_{n/2}^{\dagger}\otimes I_2]\nonumber\\
&= \frac{4}{n^2}\sum_{k,k'=0}^{n/2-1}  \sum_{(c_s,c_p)\in S}\sum_{j_1,l_1=0}^{n/2-1}\sum_{j_2,l_2=0}^{n/2-1} e^{-4\pi i (k-k') c_s/n}\left(e^{4\pi ij_1l_1}\ket{j_1}\bra{l_1}\otimes I_2\right)\nonumber\\
&\quad\times\left(\ket{k}\bra{k'} \otimes [M(k)]^t\ket{c_p}\bra{c_p} [M(k')]^{\dagger t}\right)\left(e^{-4\pi ij_2l_2}\ket{l_2}\bra{j_2}\otimes I_2\right)\nonumber\\
&=\frac{4}{n^2}\sum_{k,k'=0}^{n/2-1}  \sum_{(c_s,c_p)\in S}\sum_{j,l=0}^{n/2-1} e^{-4\pi i (k-k') c_s/n}e^{4\pi i(jk-lk')/n}\ket{j}\bra{l}\otimes [M(k)]^t\ket{c_p}\bra{c_p} [M(k')]^{\dagger t}.
\end{align}
Thus, for an arbitrary location $(s,p)$, we have

\begin{equation}
(\mathcal{A}^{t}P_S \mathcal{A}^{t\dagger})_{(j_s,j_p),(j_s,j_p)} = \frac{4}{n^2}\sum_{k,k'=0}^{n/2-1}  \sum_{(c_s,c_p)\in S} e^{-4\pi i (k-k') c_s/n}e^{4\pi i j_s(k-k')/n} \bra{j_p}[M(k)]^t\ket{c_p}\bra{c_p} [M(k')]^{\dagger t} \ket{j_p}.   
\end{equation}

\end{document}